\documentclass[sigconf]{acmart}
\usepackage[nodisplayskipstretch]{setspace}

\usepackage{graphicx}
\usepackage{mathrsfs}
\usepackage{gensymb}
\usepackage{multirow}

\usepackage{enumitem}
\setlist{leftmargin=10pt}
\AtBeginDocument{%
  \providecommand\BibTeX{{%
    \normalfont B\kern-0.5em{\scshape i\kern-0.25em b}\kern-0.8em\TeX}}}
\allowdisplaybreaks
\setcopyright{acmcopyright}
\copyrightyear{2023}
\acmYear{2023}
\acmDOI{XXXXXXX.XXXXXXX}

\acmPrice{15.00}
\acmISBN{978-1-4503-XXXX-X/18/06}

\usepackage{mathtools}

\usepackage{hyperref}
\usepackage{algorithmic}
\usepackage[ruled,linesnumbered]{algorithm2e}
\usepackage{bm}
\def\BibTeX{{\rm B\kern-.05em{\sc i\kern-.025em b}\kern-.08em
    T\kern-.1667em\lower.7ex\hbox{E}\kern-.125emX}}
\DeclareMathOperator{\secsec}{sec}
\usepackage{pifont}
\usepackage{amsmath,amsfonts,amsthm}
\begin{document}
\setlength{\abovedisplayskip}{0pt plus 0pt minus 0pt}
\setlength{\belowdisplayskip}{0pt plus 0pt minus 0pt}
\setlength\abovedisplayshortskip{0pt plus 0pt minus 0pt}
\setlength\belowdisplayshortskip{0pt plus 0pt minus 0pt}
\title{Optimizing \mbox{Utility-Energy} Efficiency  for the Metaverse over Wireless Networks under Physical Layer Security\\[-1pt]} 

\newtheorem{condition}{Condition}
\newtheorem{remark}{Remark}
\newtheorem{definitionx}{Definition}
 \newtheorem{theoremx}{Theorem}

\newenvironment{talign}
 {\let\displaystyle\textstyle\align}
 {\endalign}
\newcommand{\junzhao}[1]{\iffalse\ding{110}\ding{43}\fi\textcolor{red}{Jun Zhao: #1}}


\author{Jun Zhao, Xinyu Zhou, Yang Li, Liangxin Qian}
\affiliation{%
  \institution{Nanyang Technological University, Singapore}
  \country{junzhao@ntu.edu.sg, \{xinyu003, yang048, qian0080\}@e.ntu.edu.sg\\[-5pt]}
  }
\renewcommand{\shortauthors}{}

\begin{abstract}
The Metaverse, an emerging digital space, is expected to offer various services mirroring the real world. Wireless communications for mobile Metaverse users should be tailored to meet the following user characteristics: 1) emphasizing application-specific perceptual utility instead of simply the transmission rate, 2) concerned with energy efficiency due to the limited device battery and energy intensiveness of some applications, and 3) caring about security as the applications may involve sensitive personal data. To this end, this paper incorporates application-specific utility, energy efficiency, and physical-layer security (PLS) into the studied optimization in a wireless network for the Metaverse. 
Specifically, after introducing utility-energy efficiency (UEE) to represent each Metaverse user's application-specific objective under PLS, we formulate an optimization to maximize the network's weighted sum-UEE by deciding users' transmission powers and communication bandwidths. The formulated problem belongs to the sum-of-ratios optimization, for which prior studies have demonstrated its difficulty. Nevertheless, our proposed algorithm 1) obtains the global optimum for the weighted sum-UEE optimization, via a transform to parametric convex optimization problems, 2) applies to any utility function which is concave,  increasing, and twice differentiable, and 3) achieves a linear time complexity in the number of users (the optimal complexity in the order sense). Simulations confirm the superiority of our algorithm over other approaches. We explain that our technique for solving the  sum-of-ratios optimization is applicable to other optimization problems in wireless networks and mobile computing.




\end{abstract}




\maketitle
\thispagestyle{plain}
\pagestyle{plain}

\section{Introduction}

The Metaverse~\cite{wang2022survey} is regarded as the next generation of the Internet, which consolidates technologies including extended reality (XR), digital twin, and wireless communications. 
In 2021, Facebook changed its name to Meta, raising public interest in the Metaverse.

Mobile users typically access the Metaverse via wireless communications. It is important to optimize wireless networks to meet the attributes of Metaverse users, which we present next. 
  

\textbf{Characteristics of Metaverse users}. We identify the following traits for mobile users of the Metaverse. 
\begin{itemize}[noitemsep,topsep=0pt]
\item[\ding{172}] Users aim to maximize \textbf{application-specific perceptual utility} rather than simply the transmission rate. Traditional network optimization considers the Quality of Service (QoS), such as the transmission rate, which quantifies the objective performance of the system. For the Metaverse, humans are the main players, so  the Quality of Experience (QoE) capturing the perceptual experience of users is a better metric than QoS. To this end, our utility model should be adjusted accordingly.
\item[\ding{173}] Users care about \textbf{energy efficiency} due to the limited battery of mobile devices and energy intensiveness of some  applications. For instance, Meta Quest 2 with a fully charged battery can last for just 2 hours for gaming or 3 hours for video watching~\cite{chung2022xr}.
\item[\ding{174}] Users are concerned with \textbf{security} since certain Metaverse applications may involve  personal (e.g., biometric and health) data. Researchers at UC Berkeley have shown in~\cite{nair2022exploring} that many existing Metaverse applications are vulnerable to privacy breaches by an attacker who tries to infer users' sensitive information. 
\end{itemize} 

\textbf{The Metaverse over wireless networks: Utility-energy efficiency optimization under physical-layer security.}
Since mobile users accessing the Metaverse are constrained by wireless communication resources, it is vital to tailor wireless networks to match the above characteristics of Metaverse users. We formalize an optimization problem about the utility-energy efficiency (UEE) under physical-layer security for the motivation discussed below, where UEE for each user is defined as the application-specific perceptual utility over energy consumption. 

Energy efficiency (EE) plays a vital role in both the economy and the environment. A faster transmission rate providing a higher quality of experience for users will also increase energy consumption. Therefore, it is essential to build an energy-efficient Metaverse system.
Nevertheless, it is not viable to emphasize energy saving overwhelmingly. The Metaverse will provide many digital services, and lower transmission speeds will affect users' access to profits and high-quality experiences.  Hence, how to allocate the resources (e.g., the transmission power and bandwidth) in the  network to maximize the weighted sum of all users' UEE is worth investigating, where each user's weight represents its priority in the optimization.
The weighted sum-UEE optimization aims to save energy and improve the utilities for users, addressing ``\ding{172}'' and ``\ding{173}'' above.

For ``\ding{174}'' above,
 the confidential data of Metaverse applications should be accessible to only the intended users instead of eavesdroppers. To this end, we aim to achieve physical-layer security to protect the information during transmission. Secrecy capacity is an important metric in physical-layer security.
It refers to the communication rate that does not leak information to an eavesdropper. 
In order to keep the information of users from the eavesdroppers, we extend our Metaverse energy efficiency problem to physical-layer security by considering the secrecy rate instead of the original rate.

Our \textbf{contributions} include problem formulation, a widely applicable optimization technique, and an optimal algorithm in terms of the solution quality and time-complexity order, as listed below.
\begin{itemize}
\item We \textbf{formulate the problem} of maximizing the weighted sum of users' utility-energy efficiency (PLS) under physical-layer security for the Metaverse, by deciding users' transmission powers and bandwidth allocation. To the best of our knowledge, this problem has not yet been studied in the literature, inside and outside of Metaverse research.
\item The formulated problem belongs to the sum-of-ratios optimization, which is \mbox{non-convex}. 
We explain that the problem is difficult to solve even using the  pseudoconcavity notion. 
\item Despite the challenges, we solve the problem and develop an algorithm, via the \textbf{technique} of transforming the sum-of-ratios to parametric convex optimization problems.
\item Our proposed \textbf{algorithm} 
\begin{itemize}
\item[\textbullet] obtains the \textbf{global optimum},
\item[\textbullet]  applies to \textbf{any} utility function which is concave,  increasing, and twice differentiable, and 
\item[\textbullet] allows \textbf{heterogeneous}  utility-function types among the users,
\item[\textbullet] runs in \textbf{linear} time with respect to the number of users, which means the \textbf{optimal complexity} in the order sense.
\end{itemize}
\item 
Simulations demonstrate the superiority of our algorithm over
 other approaches in terms of the solution quality and time complexity. The utility functions used in the simulations are based on real-world datasets.
\item We explain that our \textbf{technique} can go beyond our problem to handle functions of product or quotient terms in \textbf{general} mathematical optimization. We 
illustrate this by discussing example
problems in wireless networks and mobile computing. Researchers can use our technique to solve difficult problems.
\end{itemize}



\textbf{Roadmap.} The rest of the paper is organized as follows. Section~\ref{sec:literature_review} provides related studies. In Section \ref{sec:sys_model}, we formulate the studied optimization problem. Section \ref{sec:prob_formu} presents the challenges in solving the problem. Section~\ref{secAlgorithm} elaborates on our algorithm which finds~a global optimum of the problem. 
  In Section~\ref{Insights}, we discuss  the application of our optimization technique to other problems. Simulation results are reported in Section \ref{sec:experiments}. Section \ref{sec:conclusion} concludes the paper.\vspace{-1pt} 
 

\section{Related Work\vspace{-1pt}} \label{sec:literature_review}

We survey related research: energy efficiency and physical-layer security in Section~\ref{Related work:A}, and wireless Metaverse in Section \ref{Related work:B}.\vspace{-2pt}


\subsection{Energy efficiency optimization and \mbox{physical-layer security in wireless\vspace{-1pt} networks}}\label{Related work:A}


In wireless networks, the traditional notion of energy efficiency (EE) for a user
 is defined as the ratio of data rate over power consumption (i.e., the ratio of transmitted data size over energy consumption). Maximizing the weighted sum of EE (WSEE) is addressed in~\cite{zamani2020optimizing,wu2016user}. 
Different from WSEE, the system EE  in~\cite{du2022weighted} is defined as the ratio of all users' sum rates over all users' sum power consumption. 




EE, WSEE, and system EE above do not examine specific application requirements. Accommodating various applications requires the concept of utility-energy efficiency (UEE), which for a user is the ratio of the application-specific rate-dependent utility over power consumption. UEE in our paper has also been investigated in~\cite{meshkati2009energy}, which adopts game theory to model an interference-constrained wireless network, where each user maximizes its own UEE by deciding its transmission power. Different from UEE defined for individual users, the system UEE in~\cite{huang2018utility} results from dividing the sum of all users' rate-dependent utilities by the sum of all users' power consumption. This system UEE optimization in~\cite{huang2018utility} deals with just one ratio, which is much easier than the sum-of-ratios optimization in our paper.
Moreover, the optimization method of~\cite{huang2018utility} is applicable to only the specific utility function $\kappa_n \ln r_n$ for data rate $r_n$ and constant $\kappa_n$. Even just changing the utility function to $\kappa_n \ln (1+r_n)$ will make~\cite{huang2018utility}'s approach invalid; in particular, (18a) in~\cite{huang2018utility} will be \mbox{non-concave} and \mbox{non-convex} after the above change. In contrast, our work applies to any utility function that is \mbox{concave},  increasing, and twice differentiable. Besides the above major differences,~\cite{huang2018utility} considers interference-constrained wireless networks and optimizes only the transmission powers, while we adopt FDMA and jointly optimize the transmission powers and bandwidth allocation. 

Next, we discuss the incorporation of physical-layer security (PLS) into EE optimization.
Because the WSEE as the sum of ratios is more difficult to analyze  than the system EE, existing studies incorporating PLS into EE typically investigate the system EE instead of WSEE, after replacing the achievable rates with secrecy rates, as shown in~\cite{jiang2023secrecy,zappone2019secrecy}. Despite the above work on EE optimization under PLS, we are unable to find any prior work on UEE optimization under PLS and hence the problem of our paper is new. \vspace{-1pt}




\subsection{Metaverse over wireless networks\vspace{-1pt}}\label{Related work:B}


Calibrating wireless networks for mobile users accessing the Metaverse is an emerging research topic. 
Recently, a number of papers on the topic have appeared in different venues:~\cite{yu2022asynchronous} in JSAC co-authored by the first author of the current paper,~\cite{meng2022sampling} in JSAC,~\cite{wang2023semantic} in TWC,~\cite{jiang2022reliable,ren2022quantum} in TVT, and a survey paper~\cite{wang2022survey} in COMST, where the meanings of the abbreviations can be found in the references.

Among the technical papers above,~\cite{yu2022asynchronous,meng2022sampling,ren2022quantum} adopt reinforcement learning to optimize wireless performance for the Metaverse, while~\cite{wang2023semantic,jiang2022reliable} 
utilize economic theories to incentivize users for improving the usage of semantic-aware sensing and coded distributed computing for wireless Metaverse. The current paper's co-authors have recently optimized wireless federated learning in~\cite{zhou2022resource} for the Metaverse via alternating optimization, which achieves neither local nor global optimum. In contrast, our technique of the current paper goes beyond UEE optimization under PLS. Using it in~\cite{zhou2022resource} will obtain a global optimum, based on our Section~\ref{Insights} later. \vspace{-1pt}

\section{Problem Formulation\vspace{-1pt}} \label{sec:sys_model}

In this section, we will present the system model and formalize the optimization problem.\vspace{-1pt}

\subsection{System model and metrics\vspace{-1pt}}

In our studied system, a base station acts as the Metaverse server for $N$ legitimate users $U_n|_{n=1,\ldots,N}$. There are also $N$ eavesdroppers $E_n|_{n=1,\ldots,N}$, where $E_n$ tries to intercept the communication between $U_n$ and the server. Fig.~\ref{fig:sys_model} illustrates our system.


\begin{figure}[!t]
    \centering
    \includegraphics[width=8.5cm]{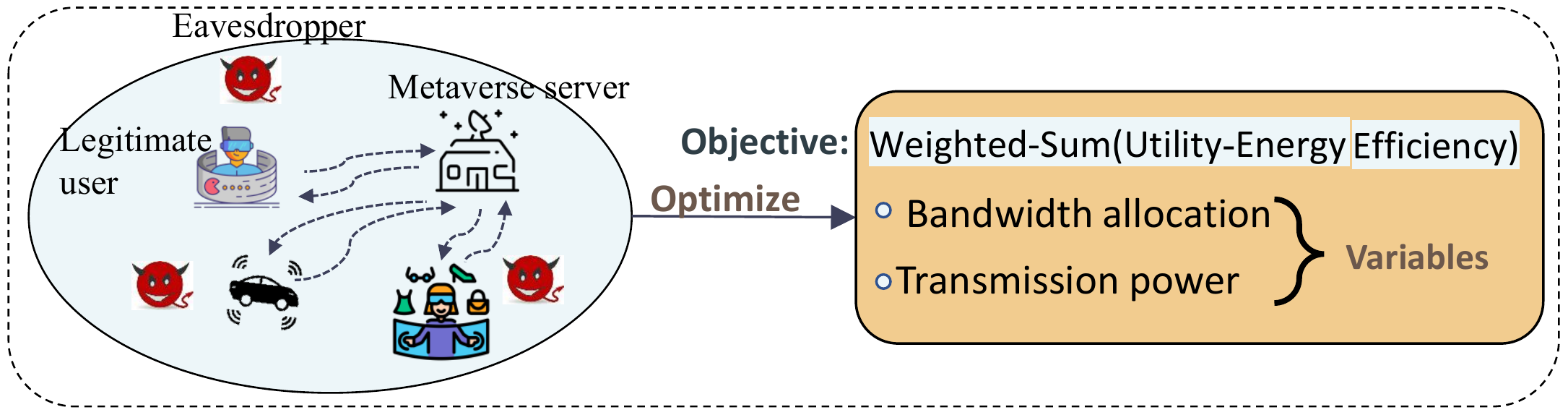}
    \vspace{-23pt}\caption{Our system: A server provides Metaverse services for $N$ legitimate users $U_n|_{n=1,\ldots,N}$, while the $n$th eavesdropper $E_n$ tries to intercept the communication between user $U_n$ and the server. The studied problem is to maximize the weighted sum of all users' utility-energy efficiency by deciding the bandwidth allocation and transmission powers. \vspace{-17pt}}
    \label{fig:sys_model}
\end{figure}

Our problem is applicable to  downlink and uplink communications between all legitimate users and the Metaverse server. \mbox{Suppose}  the communications follow frequency division multiple access (FDMA), where different legitimate users' signals will not interfere with each other. For each legitimate user $U_n$, let $B_n$ be the bandwidth, and $p_n$ be its transmission power in the case of uplink communication, or the transmission power of the server used to communicate with $U_n$ in the case of downlink communication. For simplicity, below we use uplink communication to introduce the problem.
Throughout this paper, the $n$th dimension of an $N$-dimensional vector $\bm{x}$ is denoted by $x_n$ (unless stated otherwise). 
Hence, we have $\boldsymbol p :=  [p_1, p_2, \ldots, p_N]$ and $\boldsymbol B : = [B_1, B_2, \ldots, B_N]$.


\textbf{Transmission rate.} 
According to the Shannon formula,  the transmission rate $r_n(p_n,B_n)$ of legitimate user $U_n$ is  
\begin{align}
    r_n(p_n,B_n) = B_n\log_2(1+\frac{g_n p_n}{{\sigma_n}^2B_n}),
\end{align}
where ${\sigma_n}^2$ is the power spectral density of Gaussian noise, $g_n$ is the channel attenuation from $U_n$ to the server. 
The function notation is used in this paper; e.g., $r_n(p_n,B_n)$ is a function of $p_n$ and $B_n$.






\textbf{Secrecy rate.} Eavesdropper $E_n$ aims to intercept the communication between legitimate user $U_n$ and the server. Let $r_{n,\textnormal{e}}$ be the eavesdropping rate of $E_n$. We consider  $r_{n,\textnormal{e}}$ as a constant depending on only $n$. Then the secrecy rate of $U_n$ is given by
\begin{align}
r_{n,\textnormal{s}}(p_n,B_n):=r_n(p_n,B_n) - r_{n,\textnormal{e}}. \label{defrns}
\end{align}

\textbf{Utility.}
We regard user $U_n$'s application-specific perceptual utility rate as a function of the secrecy rate $r_{n,\textnormal{s}}(p_n,B_n)$ to emphasize physical-layer security (PLS). Specifically, using\footnote{We require $f_n(x)$ to be defined for any $x >0$. We do not require $f_n(x)$ to be defined for $x=0$, but if $\lim_{x\to 0^+}f_n(x)$ exists and is finite, we can just use it to define $f_n(0)$. We also do not enforce any condition on whether $f_n(x)$  is \mbox{non-negative} or not. Additional conditions of $f_n(x)$ are discussed in Section~\ref{secfConditions}.\label{ftntsecfConditions}} $f_n(\cdot): (0, \infty) \to  (-\infty, \infty)$ to denote the utility rate function, user $U_n$'s utility rate is given by $f_n(r_{n,\textnormal{s}}(p_n, B_n))$. Consider a small time interval $[t, t+\Delta t]$, where $\Delta t$ is small enough such that $r_{n,\textnormal{s}}(p_n,B_n)$ can be seen as \mbox{invariant}  during  $[t, t+\Delta t]$. Then the utility of user $U_n$ over the time interval $[t, t+\Delta t]$ is $\mathcal{U}_n^{[t, t+\Delta t]}:=f_n(r_{n,\textnormal{s}}(p_n, B_n)) \Delta t$.

\textbf{Power \& energy 
consumption.} The same as~\cite{huang2018utility,xu2013throughput}, the power consumed by user $U_n$ includes not just the transmission power $p_n$, but also the circuit power $p_n^{\textnormal{cir}}$, which is  a constant given $n$. During the time interval $[t, t+\Delta t]$, user $U_n$'s energy consumption is given by $\mathcal{E}_n^{[t, t+\Delta t]}:=(p_n + p_n^{\textnormal{cir}}) \Delta t$.



\textbf{\mbox{Utility-energy} efficiency.} For user $U_n$, we define its utility-energy efficiency (UEE) $\varphi_n(p_n,B_n)$ under PLS as the ratio of $\frac{\mathcal{U}_n^{[t, t+\Delta t]}}{\mathcal{E}_n^{[t, t+\Delta t]}}$ for small enough $\Delta t$, which induces
\begin{align} 
\varphi_n(p_n,\hspace{-1pt}B_n)\hspace{-1pt}:= \hspace{-1pt} \frac{f_n(r_{n,\textnormal{s}}(p_n,\hspace{-1pt}B_n))}{p_n + p_n^{\textnormal{cir}}} \hspace{-1pt}=\hspace{-1pt}  \frac{f_n(r_n(p_n,\hspace{-1pt}B_n)\hspace{-1pt}-\hspace{-1pt}r_{n,\textnormal{e}})}{p_n + p_n^{\textnormal{cir}}} .\label{eqdefvarphi}
\end{align}
When $f_n(\cdot)$ becomes the identity function (i.e., $f_n(x)=x$), $\varphi_n(p_n,B_n)$ becomes $\frac{r_{n,\textnormal{s}}(p_n, B_n)}{p_n + p_n^{\textnormal{cir}}}$ (i.e., $\frac{\text{secrecy rate}}{\text{power consumption}}$), which is just the traditional notion of energy efficiency under PLS~\cite{jiang2018secrecy}.

 \subsection{Utility-energy efficiency (UEE) optimization} \label{optproblem}














Our goal is to maximize the weighted sum of all users' UEE under PLS. This optimization problem is formalized as follows:
\begin{subequations}\label{problem:1}
\begin{align} 
   \textnormal{Problem $\mathbb{P}_1$:~}  &\max_{\boldsymbol{p}, \boldsymbol{B}} ~
   \sum_{n \in \mathcal{N}}  c_n\varphi_n(p_n,B_n)
    \tag{\ref{problem:1}}\\ 
    \textnormal{subject to: }    & \sum_{n \in \mathcal{N}} B_n \le B_{\textnormal{total}},
    \label{constra:bandwidth} \\ 
   & r_n(p_n,B_n) \geq r_n^{\min}, \text{ for all } n \in \mathcal{N} := \{1,\cdots, N\}, \label{constra:rate}
\end{align}
\end{subequations}
where
$c_n > 0$ represents the priority of user $U_n$ in the optimization. Larger $c_n$ means higher priority. Constraints \text{(\ref{constra:bandwidth})} sets the total bandwidth for FDMA. Constraint (\ref{constra:rate}) ensures that the  transmission rate $r_n(p_n,B_n)$ of  user $U_n$ should be at least a constant $r_n^{\min}$ ($r_n^{\min}$ can vary for different $n$). Condition~\ref{conditionrn} below is about minimum legitimate rates $r_n^{\min}|_{n \in \mathcal{N}} $ and eavesdropping rates $ r_{n,\textnormal{e}}|_{n \in \mathcal{N}} $.\vspace{-1pt}



\begin{condition} \label{conditionrn}
\hspace{-3pt}For all $n \in \mathcal{N}$, we have \mbox{$r_n^{\min} \hspace{-1pt}\geq\hspace{-1pt} r_{n,\textnormal{e}}, r_n^{\min} \hspace{-1pt}>\hspace{-1pt}0, r_{n,\textnormal{e}} \hspace{-1pt}\geq \hspace{-1pt}0$}.  
\end{condition}

We have the following remarks about Condition~\ref{conditionrn}.

\begin{remark} \label{remarkrn0}
Condition~\ref{conditionrn} with~(\ref{constra:rate}) ensures \mbox{$r_n(p_n,B_n) \geq r_{n,\textnormal{e}}$}; i.e., each user $U_n$'s secrecy rate $r_{n,\textnormal{s}}(p_n,B_n)$ 
is \textit{non-negative}.
\end{remark}


\begin{remark} \label{remarkrn1}
Condition~\ref{conditionrn} covers the following special case where we do not impose physical-layer security but still enforce a minimum transmission rate for each user: $r_{n,\textnormal{e}} = 0$ and $r_n^{\min} >0$ for all $n \in \mathcal{N}$.
\end{remark}

\begin{remark} \label{remarkrn2}
We enforce $r_n^{\min} >0$ in Condition~\ref{conditionrn} so that each user $U_n$ will always be allocated with a strictly positive bandwidth; i.e., $B_n >0$ for all $n \in \mathcal{N}$. This avoids analyzing the degenerate case where only a subset of $N$ users share the total bandwidth \vspace{-1pt}$B_{\textnormal{total}}$. 

\end{remark}

 We also comment on how Problem $\mathbb{P}_{1}$ is optimized in\vspace{-1pt} practice.

\begin{remark}
Problem $\mathbb{P}_{1}$ will be solved using our Algorithm~\ref{algo:MN} in Section~\ref{secproveP1v2}. Then a  question  is which entity solves $\mathbb{P}_{1}$ in practice. We let the Metaverse server perform the task, assuming that it has obtained the values of $p_n^{\textnormal{cir}},r_n^{\min}, r_{n,\textnormal{e}}$ for all $n$ (e.g., these are shared with the server before the optimization stage). After the server solves $\mathbb{P}_{1}$, it will notify each legitimate user $U_n$ of the $p_n$ and $B_n$\vspace{-1pt} values.
\end{remark}
 




\section{Challenges in Solving Problem $\mathbb{P}_1$\vspace{-1pt}} \label{sec:prob_formu}

We first state the optimization preliminaries and conditions of the function $f_n(x)$, which are used to explain the difficulty in solving~$\mathbb{P}_1$.








\subsection{\mbox{Preliminaries of mathematical optimization}}

Let $f(\boldsymbol{x})$ be a function defined on a convex set $\mathcal{S}$, which is a subset of a real vector space. Then we have the following from Definitions 1.3.1, 2.2.1, and 3.2.1 of the book~\cite{cambini2008generalized}.
\begin{definitionx}[Convexity]
$f$ is convex in $\boldsymbol{x}$ if and only if for any $\boldsymbol{x}_1, \boldsymbol{x}_2 \in \mathcal{S}$ and $t \in [0,1]$, it holds that \\$f(t\boldsymbol{x}_1 + (1-t)\boldsymbol{x}_2) \leq t f(\boldsymbol{x}_1) + (1-t) f(\boldsymbol{x}_2)$.
\end{definitionx}

\begin{definitionx}[Pseudoconvexity]
$f$ is pseudoconvex in $\boldsymbol{x}$ if and only if for any $\boldsymbol{x}_1, \boldsymbol{x}_2 \in \mathcal{S}$, $f(\boldsymbol{x}_1) > f(\boldsymbol{x}_2)$ implies \\$\nabla f(\boldsymbol{x}_1) \cdot (\boldsymbol{x}_2-\boldsymbol{x}_1) < 0$, where $\nabla f$ denotes the gradient of $f$.
\end{definitionx}


\begin{definitionx}[Quasiconvexity]
$f$ is quasiconvex in $\boldsymbol{x}$ if and only if for any $\boldsymbol{x}_1, \boldsymbol{x}_2 \in \mathcal{S}$ and $t \in [0,1]$, it holds that\\ $f(t\boldsymbol{x}_1 + (1-t)\boldsymbol{x}_2) \leq \max\{ f(\boldsymbol{x}_1) , f(\boldsymbol{x}_2)\}$.
\end{definitionx}

With convexity above, Lemma~\ref{lemmaConvexityvsConcavity} helps us understand concavity.

\begin{lemma}[Convexity versus Concavity] \label{lemmaConvexityvsConcavity}
A function $f$ is said to be concave (resp., pseudoconcave, quasiconcave) if and only
if $-f$ is convex (resp., pseudoconvex, quasiconvex). 
\end{lemma}

For the reasoning behind Lemma~\ref{lemmaConvexityvsConcavity}, interested readers can refer to Section 3 of the book~\cite{cambini2008generalized}.
Lemma~\ref{ConvexityRelationships} below presents the relationships between the definitions discussed above.

\begin{lemma}[Relationships between notions]
\label{ConvexityRelationships}
With ``$\Rightarrow$'' denoting ``implies'',  we have the following assuming differentiability
\begin{align}
&\text{Convexity} \Rightarrow \text{Pseudoconvexity}  
\Rightarrow \text{Quasiconvexity} ,  \text{ and} 
\\
&\text{Concavity} \Rightarrow \text{Pseudoconcavity}  
\Rightarrow \text{Quasiconcavity}.    
\end{align} 
\end{lemma}

Lemma~\ref{ConvexityRelationships} follows from Fig.~2.2 and Fig.~B.1 of the book~\cite{cambini2008generalized}.


For a minimization problem, if the objective function and\footnote{Note that the constraint of a convex (resp., concave)  function being at most  (resp., least) a constant is a convex constraint.\label{convexconstraint}} the constraints are all convex, then we have a convex optimization problem, for which the following lemma holds.

\begin{lemma}[Chapters 3 and 4 of~\cite{boyd2004convex}]\label{lemmacvxuseKKT}
For convex optimization, the Karush--Kuhn--Tucker (KKT) conditions are
\begin{itemize}
    \item sufficient for optimality, and
        \item are necessary for optimality if Slater's condition holds (i.e., if the feasible set contains at least one interior point).
\end{itemize} 
\end{lemma}

Readers unfamiliar with the KKT conditions can refer to Theorem 4.2.3 of~\cite{cambini2008generalized}, and can also look into (\ref{P2KKTeqStationaritypn})--(\ref{P2KKTeqDualfeasibility}) to be presented on Page~\pageref{P2KKTeqStationaritypn} of the current paper, where we will use the KKT conditions. 

Lemma~\ref{lemmauseKKT} below broadens problems under which KKT conditions are sufficient for optimality, to go beyond convex optimization.



\begin{lemma}[Theorem 4.4.1 of~\cite{cambini2008generalized}]\label{lemmauseKKT}
For a minimization problem with all constraints being inequalities, if the objective function is pseudoconvex, and all constraints are quasiconvex and differentiable, then a feasible point satisfying the KKT conditions is globally optimal.     
\end{lemma}

\subsection[Conditions of the utility function fn(x)]{Conditions of the utility function $f_n(x)$} \label{secfConditions}

The requirements of  the utility rate function $f_n(x)$
for any $n \in \mathcal{N}$ are formally presented as Condition~\ref{fconcave} below (we will just call $f_n(x)$ as the utility function hereafter for simplicity). 






\begin{condition}\label{fconcave}
The utility function $f_n(x)$ for any $n \in \mathcal{N}$ is \mbox{concave},  increasing, and twice differentiable, with respect to $x > 0$; i.e., \mbox{$f_n''(x) \leq 0$} and $f_n'(x)>0$ for $x > 0$.
\end{condition}

In Condition~\ref{fconcave}, the concavity of $f_n(x)$ means diminishing marginal return, which  holds in various practical applications~\cite{yang2012crowdsourcing,liu2018edge,mo2000fair,xiong2020reward}.


We also remark that $f_n(x)$, $f_n'(x)$ and $f_n''(x)$ are defined for any $x >0$. Similar to Footnote~\ref{ftntsecfConditions}, for any of $f_n(x)$, $f_n'(x)$ and $f_n''(x)$, we do not require it to be defined for $x=0$. If it has a finite (resp., no) limit as $x \to 0^+$, we can just use the limit to define the corresponding value at $x = 0$ (resp., do not define any value at $x = 0$). 


We can even allow heterogeneous types of utility functions among the users. For the specific expressions of the  utility function $f_n(x)$ used in simulations,
we will discuss three types in Section~\ref{secTypicalUtility}. 

We now discuss the properties of functions in our\vspace{-1pt} studied system.  

\begin{lemma}[Lemma 1 of~\cite{zhou2022resource}]\label{lemma:fn_concave} \label{lemma:rn_concave}
  $r_n(p_n,B_n) $ is  jointly concave with\footnote{For a function $f(\boldsymbol{x})$, ``being convex (resp., concave) in $\boldsymbol{x}$'' has the same meaning as ``jointly convex (resp., concave) in all dimensions of the vector $\boldsymbol{x}$''.} respect to $p_n$ and $B_n$.    \vspace{-1pt} 
\end{lemma}

\begin{lemma}\label{lemma:fn_concave}
Under Condition~\ref{fconcave}, we have:
\begin{itemize}
    \item $f_n(r_{n,\textnormal{s}}(p_n,B_n))$ is jointly concave with respect to $p_n$ and $B_n$;
        \item $\varphi_n(p_n,B_n)$ is jointly pseudoconcave  with respect to $p_n$ and $B_n$.\vspace{-1pt} 
\end{itemize}  
\end{lemma}
\begin{proof}
From Lemma~\ref{lemma:rn_concave} and Eq.~(\ref{defrns}), $r_{n,\textnormal{s}}(p_n,B_n)$ is jointly concave in $p_n$ and $B_n$. According to the composition rule in Eq.~(3.11) of~\cite{boyd2004convex}, for concave $f_n(\cdot)$, since $\widetilde{f_n}(\cdot)$ defined as $f_n(\cdot)$ (resp., $-\infty$) for points inside (resp., outside) of the domain of $f_n$ is \mbox{non-decreasing}, the function $f_n(r_{n,\textnormal{s}}(p_n, B_n))$ is jointly concave in $p_n$ and $B_n$.  



From Page 245 (the book's internal page number, not the pdf page number) of the book~\cite{cambini2008generalized}, for a ratio, if the numerator is \mbox{non-negative}, concave and differentiable, and the denominator is positive, convex and differentiable, then the ratio is pseudoconcave. Based on the above, we have proved the pseudoconcavity of $\varphi_n(p_n,B_n)$ with respect to $p_n$ and $B_n$.\vspace{-1pt} 
\end{proof}


\subsection[]{Challenges of solving Problem $\mathbb{P}_1$\vspace{-1pt} } \label{P1Challenges}

Based on the proof of Lemma~\ref{lemma:fn_concave}, we now call $\small\frac{f_n(r_{n,\textnormal{s}}(p_n,B_n))}{p_n + p_n^{\textnormal{cir}}}$ (i.e., $\varphi_n(p_n,B_n)$) a \mbox{concave-convex} ratio: a ratio having a concave function as the numerator and a convex function as the denominator. Then  $\mathbb{P}_{1}$ is maximizing the sum of  \mbox{concave-convex} ratios. Such sum-of-ratios optimization is \mbox{non-convex} and difficult to solve~\cite{jong2012efficient,shen2018fractional}. 


Lemma~\ref{lemma:fn_concave} also shows that $\varphi_n(p_n,B_n)$ for each $n$ is \mbox{pseudoconcave}, unfortunately the sum of \mbox{pseudoconcave} functions may not be \mbox{pseudoconcave}. Even if we manage to prove the pseudoconcavity of $\sum_{n \in \mathcal{N}}  c_n\varphi_n(p_n,B_n)$ (the objective function of $\mathbb{P}_1$), which is very difficult (e.g., just analyzing the pseudoconvexity of the sum of two linear fractional functions is already challenging, as shown in~\cite{cambini2005pseudoconvexity}),  then we can in principle use the KKT conditions of Problem $\mathbb{P}_{1}$, as explained in\footnote{For $\mathbb{P}_1$, all constraints are  differentiable and convex (and hence quasiconvex) with Lemma~\ref{lemma:rn_concave} and Foonote~\ref{convexconstraint}. Hence,  if we can prove that the objective function is \mbox{pseudoconcave},   Lemma~\ref{lemmauseKKT} means that KKT conditions can solve $\mathbb{P}_{1}$. Nonetheless, even if we can do the above, the KKT conditions of $\mathbb{P}_{1}$ are intractable to get a solution. \label{analyzeP1}} Footnote~\ref{analyzeP1}, but those conditions involve taking derivatives of the ratios, inducing quite complex expressions, and the corresponding analysis becomes intractable. In this paper, instead of analyzing the pseudoconcavity of $\sum_{n \in \mathcal{N}}  c_n\varphi_n(p_n,B_n)$ and being trapped in the intractable analysis, we will present an elegant approach (to be detailed in Section~\ref{secTransforming}) for solving Problem $\mathbb{P}_{1}$.

Recently, Shen and Yu~\cite{shen2018fractional,shen2018fractional2} proposed a novel technique to solve the sum-of-ratios optimization (referred to as fractional programming in their papers). However, since their technique relies on block coordinate ascent (i.e., alternating optimization), applying their technique to our Problem $\mathbb{P}_1$ will find a point which has no local or global optimality guarantee. In contrast, our approach will find a globally optimal solution of $\mathbb{P}_1$.


       

\section{\mbox{Algorithm to Find A Global Optimum}}\label{secAlgorithm} 

In this section, we will discuss how to transform $\mathbb{P}_1$ into a sequence of convex optimization problems, and then use the transform to propose an algorithm that finds a global optimum of $\mathbb{P}_1$. 


 
\subsection{Transforming Problem $\mathbb{P}_1$ into parametric convex optimization problems}\label{secTransforming}


Firstly, we introduce an auxiliary variable $\beta_n$ to transform Problem $\mathbb{P}_{1}$ into the epigraph form. Let $\frac{c_nf_n(r_{n,\textnormal{s}}(p_n,B_n))}{p_n+p_n^{\textnormal{cir}}} \ge \beta_n$ and $\mathbb{P}_{1}$ can be transformed to the following equivalent form as $\mathbb{P}_2$:
\begin{subequations}\label{problem:2}
\begin{align}
   \textnormal{Problem $\mathbb{P}_2$:}  ~ &\max_{\boldsymbol{p}, \boldsymbol{B},  \boldsymbol{\beta}} ~
   \sum_{n \in \mathcal{N}} \beta_{n}, \text{ where }  \mathcal{N} := \{1,\cdots, N\}, 
    \tag{\ref{problem:2}}\\ 
    \textnormal{subject to: } & \text{(\ref{constra:bandwidth}), (\ref{constra:rate})},    \\
    & \hspace{-30pt} F_n(p_n, B_n) -\beta_{n}\cdot(p_n + p_n^{\textnormal{cir}}) \geq 0, \textnormal{ for all $n \in \mathcal{N}$},\label{constra:pysical_epi_data_rate}
\end{align}
\end{subequations}
 where we   use $F_n(p_n, B_n)$ to simplify the representation:
\begin{align}
  & F_n(p_n, B_n) : = c_n f_n(r_{n,\textnormal{s}}(p_n,B_n))  . \label{defineFnpnBn}
    \end{align}
Problem $\mathbb{P}_2$ is not convex optimization since \mbox{$\beta_{n}\cdot(p_n + p_n^{\textnormal{cir}})$}  in~(\ref{constra:pysical_epi_data_rate}) is not jointly convex (actually also not jointly concave) in $\beta_{n}$ and $p_n$, since the Hessian matrix for \mbox{$\beta_{n}\cdot(p_n + p_n^{\textnormal{cir}})$} is $\tiny\begin{bmatrix}
0 & 1 \\ 1 & 0    
\end{bmatrix}$ which is not positive semifinite (actually also not negative semidefinite).

We have explained in Section~\ref{P1Challenges} that Problem $\mathbb{P}_{1}$ belongs to the following kind of problems: maximizing the sum of \mbox{concave-convex} ratios. Such problems have at least one global maximum according   to~\cite{jong2012efficient,shen2018fractional}. Hence, $\mathbb{P}_{1}$ and  $\mathbb{P}_{2}$ have at least one global maximum. 


To solve Problem $\mathbb{P}_2$, one initial idea is trying to use Lemma~\ref{lemmauseKKT} and hence the KKT conditions directly. Yet, deciding the quasiconvexity of \mbox{$F_n(p_n, B_n) -\beta_{n}\cdot(p_n + p_n^{\textnormal{cir}})$}   in (\ref{constra:pysical_epi_data_rate}) is very difficult. Hence, instead of trying to use $\mathbb{P}_2$'s KKT conditions directly, we take a step back and use the Fritz-John conditions (viz., Remark~4.2.2 of~\cite{cambini2008generalized} and Lemma~2.1's proof in~\cite{jong2012efficient}), which do not need the quasiconvexity of constraints. The Fritz-John conditions provide the necessary conditions for a global optimum. Basically, in the Fritz-John conditions, the Lagrange multiplier (say $w$) on the gradient of the objective function can be zero or positive. Yet, following the proof of Lemma 2.1 in~\cite{jong2012efficient}, we obtain $w>0$. Then as shown in~\cite{jong2012efficient}, $w$ can be absorbed into other multipliers and hence omitted, after which the Fritz-John conditions reduce to the KKT conditions. Based on the above discussion, we have
\begin{talign}
\begin{array}{l}
\text{any global maximum of $\mathbb{P}_{2}$ needs to satisfy}\\ \text{the KKT conditions (\ref{P2KKTeqStationaritypn})--(\ref{P2KKTeqDualfeasibility}) below.}
\end{array} 
 \label{showP2KKT}    
\end{talign}

For Problem $\mathbb{P}_{2}$, with $\bm{\nu} := [\nu_n|_{n \in \mathcal{N}}]$, $\bm{\tau} := [\tau_n|_{n \in \mathcal{N}}]$ and $\lambda$ denoting the multipliers, and the Lagrangian function given by
\begin{talign} \label{LagrangianP2} 
   &  L_{\mathbb{P}_2}(\boldsymbol{p}, \boldsymbol{B},\boldsymbol{\beta},\boldsymbol{\nu},  \boldsymbol{\tau}, \lambda)   \nonumber
     \\
     &  = - \sum_{n \in \mathcal{N}} \beta_{n}
 + \sum_{n\in \mathcal{N}} \nu_n\cdot \big(\beta_n \cdot (p_n+p_n^{\textnormal{cir}}) - F_n(p_n, B_n) \big) \nonumber
    \\
     & \quad \,\textstyle{+  \sum_{n\in \mathcal{N}} \tau_n\cdot(r_n^{\min} -r_n ) + \lambda\cdot(\sum_{n\in \mathcal{N}}B_n- B_{\text{total}})},
\end{talign}
the KKT conditions of Problem $\mathbb{P}_{2}$ are as follows, with $L_{\mathbb{P}_2}$ short for $L_{\mathbb{P}_2}(\boldsymbol{p}, \boldsymbol{B},\boldsymbol{\beta},\boldsymbol{\nu},  \boldsymbol{\tau}, \lambda)  $ (see~\cite[Theorem 4.2.3]{cambini2008generalized} or~\cite[Section~1.4.2]{boyd2004convex} for a formal introduction to the KKT conditions):
\begin{subequations} \label{P2KKTeq}
\begin{talign} 
&\textnormal{\textbf{Stationarity:}} \nonumber
\\  &\frac{\partial L_{\mathbb{P}_2}}{\partial p_n}     = 0, \text{ for all } n \in \mathcal{N}, \label{P2KKTeqStationaritypn} \\ &\frac{\partial L_{\mathbb{P}_2}}{\partial B_n}     = 0, \text{ for all } n \in \mathcal{N}, \label{P2KKTeqStationarityBn} \\& \frac{\partial L_{\mathbb{P}_2}}{\partial \beta_n}  =  - 1 + \nu_n \cdot (p_n+p_n^{\textnormal{cir}})  = 0, \text{ for all } n \in \mathcal{N}; \label{P2KKTeqStationaritybetan} \\ &\textnormal{\textbf{Complementary slackness:}} \nonumber
\\ & \nu_n\cdot \big(\beta_n \cdot (p_n+p_n^{\textnormal{cir}}) - F_n(p_n, B_n) \big) = 0, \text{ for all } n \in \mathcal{N},\label{P2KKTeqComplementarynun} \\ & \tau_n\cdot(r_n^{\min} -r_n ) = 0, \text{ for all } n \in \mathcal{N}, \label{P2KKTeqComplementarytaun}  \\ &  \lambda\cdot(\sum_{n\in \mathcal{N}}B_n- B_{\text{total}}) =0; \label{P2KKTeqComplementarylambda}  \\ &\textnormal{\textbf{Primal feasibility:}} \nonumber
\\  & F_n(p_n, B_n) -\beta_{n}\cdot(p_n + p_n^{\textnormal{cir}}) \geq 0, \textnormal{ for all $n \in \mathcal{N}$} , \label{P2KKTeqPrimalfeasibility1}\\ & r_n(p_n,B_n) \geq r_n^{\min}, \text{for all } n \in \mathcal{N}, \label{P2KKTeqPrimalfeasibility2}\\  & \sum_{n\in \mathcal{N}}B_n \leq B_{\text{total}}; \label{P2KKTeqPrimalfeasibility3}\\ &\textnormal{\textbf{Dual feasibility:}} \nonumber
\\ & \nu_n \geq 0, \text{for all } n \in \mathcal{N}, \label{P2KKTeqDualfeasibilitynu}\\ &  \tau_n  \geq 0,  \text{for all } n \in \mathcal{N}. \label{P2KKTeqDualfeasibilitytau}\\ &  \lambda  \geq 0. \label{P2KKTeqDualfeasibility}  
\end{talign}
\end{subequations}

From~(\ref{P2KKTeqStationaritybetan}), it follows that
\begin{talign}
 & \nu_n = \frac{1}{p_n+p_n^{\textnormal{cir}}} , \label{P2KKTeqStationaritybetannures}
\end{talign} 
and (\ref{P2KKTeqDualfeasibilitynu}) holds.

Using~(\ref{P2KKTeqStationaritybetannures}) in~(\ref{P2KKTeqComplementarynun}), we know
\begin{talign}
 & \beta_n =    \frac{F_n(p_n, B_n)}{p_n+p_n^{\textnormal{cir}}},\label{P2KKTeqStationaritybetannuresbeta}
\end{talign}
and (\ref{P2KKTeqPrimalfeasibility1}) holds (actually the equal sign in~(\ref{P2KKTeqPrimalfeasibility1}) is taken).

Instead of solving $\mathbb{P}_{2}$'s KKT conditions~(\ref{P2KKTeqStationaritypn})--(\ref{P2KKTeqDualfeasibility}) directly, which is complex, we will connect them to a series of parametric convex optimization problems. In particular, supposing that $\boldsymbol{\beta} $ and $\boldsymbol{\nu}$ are already given and satisfy (\ref{P2KKTeqStationaritybetan})  (\ref{P2KKTeqComplementarynun}) (\ref{P2KKTeqPrimalfeasibility1}) and  (\ref{P2KKTeqDualfeasibilitynu}), then we have the following result for the rest of $\mathbb{P}_{2}$'s KKT conditions:
\begin{talign}
\hspace{-10pt}\begin{array}{l}
\text{(\ref{P2KKTeqStationaritypn}) (\ref{P2KKTeqStationarityBn})  (\ref{P2KKTeqComplementarytaun}) (\ref{P2KKTeqComplementarylambda}) (\ref{P2KKTeqPrimalfeasibility2}) (\ref{P2KKTeqPrimalfeasibility3}) (\ref{P2KKTeqDualfeasibilitytau}) and (\ref{P2KKTeqDualfeasibility}), denoted by}\\ \text{set $\mathcal{K}$, form the KKT conditions of Problem $\mathbb{P}_3(\bm{\beta},\bm{\nu})$ below.}
\end{array} 
 \label{showP3KKT}    
\end{talign}
where we have
\begin{subequations}\label{problem:3}
   \begin{align}
 &  \textnormal{Problem $\mathbb{P}_3(\bm{\beta},\bm{\nu})$:} ~  & \max_{\boldsymbol{p}, \boldsymbol{B}}\sum_{n \in \mathcal{N}}  \mathcal{F}_n(p_n, B_n \,| \, \beta_n, \nu_n)\tag{\ref{problem:3}}   \\[-5pt]
 &\textnormal{subject to: } ~ \textnormal{(\ref{constra:bandwidth}), (\ref{constra:rate})} . \notag
    \end{align}
 \end{subequations}
with $\mathcal{F}_n(p_n, B_n \,| \, \beta_n, \nu_n)$ defined as follows: 
\begin{align}
  \mathcal{F}_n(p_n, B_n \,| \, \beta_n, \nu_n) :    &= \nu_n\cdot \big(F_n(p_n, B_n) -\beta_n \cdot (p_n+p_n^{\textnormal{cir}})\big). \label{mathcalF1}
    \end{align}

Lemma~\ref{lemma:sum-of-ratios_lemma} below states the relationship between $\mathbb{P}_{2}$ and $\mathbb{P}_{3}$.

\begin{lemma}\label{lemma:sum-of-ratios_lemma}
If we have \textbf{\ding{192}:} $[\boldsymbol{p}^*, \boldsymbol{B}^*, \boldsymbol{\beta}^*]$ is a globally optimal solution to Problem $\mathbb{P}_{2}$, then we get \textbf{\ding{193}:} $\bm{\beta}^*$ denoting $[\beta_n^*|_{n \in \mathcal{N}}]$ satisfies
\begin{align}
  &\beta_n^*=\frac{F_n(p_n^*, B_n^*)}{p_n^*+p_n^{cir}},\text{ for all } n \in \mathcal{N} ,\label{decidebetanstar} 
\end{align}
and \textbf{\ding{194}:}  
 $[\boldsymbol{p}^*, \boldsymbol{B}^*]$ is a globally optimal solution to Problem $\mathbb{P}_3(\bm{\beta}^*,\bm{\nu}^*)$, where we have \textbf{\ding{195}:} $\bm{\nu}^*$ denoting $[\nu_n^*|_{n \in \mathcal{N}}]$ is
  given by 
\begin{align}
  & \nu_n^*=\frac{1}{p_n^*+p_n^{cir}}, \text{ for all } n \in \mathcal{N}.\label{decidenunstar} 
\end{align}
\end{lemma}
\begin{proof} 
Problem $\mathbb{P}_3(\bm{\beta},\bm{\nu})$ belongs to convex optimization. In particular, according to Lemma~\ref{lemma:fn_concave} and (\ref{defineFnpnBn})  (\ref{mathcalF1}), the objective function to be maximized is concave, while the constraints are clearly convex with Lemma~\ref{lemma:rn_concave} (note that ``\text{concave} $\geq$ \text{constant}'' is a convex constraint as noted in Footnote~\ref{convexconstraint}). Also, Slater's condition holds for Problem $\mathbb{P}_3(\bm{\beta},\bm{\nu})$. In other words, there exists at least one point $[\boldsymbol{p}, \boldsymbol{B}]$ such that constraints (\ref{constra:bandwidth}) and~(\ref{constra:rate}) are satisfied with strict inequalities.  An example is as follows: with $B_n$ being $\frac{B_{\text{total}}}{2N}$ for all $n \in \mathcal{N}$, set $p_n$ such that $r_n(p_n,B_n) = 2 r_n^{\min}$ for all $n \in \mathcal{N}$. The above along with Lemma~\ref{lemmacvxuseKKT} shows the first ``$\Leftrightarrow$'' result below: 
\begin{talign}
 \begin{rcases*} \text{\ding{194}} \Leftrightarrow \text{$[\boldsymbol{p}^*, \boldsymbol{B}^*, \boldsymbol{\beta}^*,\bm{\nu}^*]$ satisfies the set $\mathcal{K}$ of conditions in (\ref{showP3KKT}).} \\ \hspace{10pt} \text{Results \ding{193} and \ding{195} hold; i.e., $[\bm{\beta}^*,\bm{\nu}^*]$ satisfies (\ref{decidebetanstar}) and   (\ref{decidenunstar}).} \end{rcases*} \nonumber \\ \Leftrightarrow \text{$[\boldsymbol{p}^*, \boldsymbol{B}^*, \boldsymbol{\beta}^*,\bm{\nu}^*]$ satisfies KKT conditions (\ref{P2KKTeqStationaritypn})--(\ref{P2KKTeqDualfeasibility})}  \Leftarrow  \text{\ding{192}} 
  \label{lemallresults} 
  \end{talign}  
where the second ``$\Leftrightarrow$'' above holds from (\ref{P2KKTeqStationaritybetannures}) and   (\ref{P2KKTeqStationaritybetannuresbeta}), and the last ``$\Leftarrow$'' above follows from (\ref{showP2KKT}).
\end{proof}

For additional understanding of Lemma~\ref{lemma:sum-of-ratios_lemma}, interested readers can refer to Lemma 2.1 and Remark 2.1 of~\cite{jong2012efficient}, where Lemma 2.1 of~\cite{jong2012efficient} handles minimizing the sum of convex-concave ratios and Remark 2.1 of~\cite{jong2012efficient} maximizes the sum of \mbox{concave-convex} ratios.



 

With Lemma~\ref{lemma:sum-of-ratios_lemma} presented above, we now describe how to solve Problem $\mathbb{P}_{2}$ using $\mathbb{P}_3(\bm{\beta},\bm{\nu})$. 
Let $[\bm{p}^{\#}(\bm{\beta},\bm{\nu}), \bm{B}^{\#}(\bm{\beta},\bm{\nu})]$ denote
a globally optimal solution  to  $\mathbb{P}_3(\bm{\beta},\bm{\nu})$, where $\bm{p}^{\#}(\bm{\beta},\bm{\nu})=[p_n^{\#}(\bm{\beta},\bm{\nu})|_{n \in \mathcal{N}}]$ and $\bm{B}^{\#}(\bm{\beta},\bm{\nu})=[B_n^{\#}(\bm{\beta},\bm{\nu})|_{n \in \mathcal{N}}]$.   
We further define
\begin{talign}
\phi_{1,n}(\bm{\beta},\bm{\nu}) & \hspace{-2pt}:=\hspace{-2pt}- F_n(p_n^{\#}(\bm{\beta},\bm{\nu}), B_n^{\#}(\bm{\beta},\bm{\nu}))   + \beta_n \hspace{-2pt}\cdot\hspace{-2pt} (p_n^{\#}(\bm{\beta},\bm{\nu})\hspace{-2pt}+\hspace{-2pt}p_n^{\textnormal{cir}}),  \label{eqn:phi_1v2}
\\
\phi_{2,n}(\bm{\beta},\bm{\nu}) & \!:=\!-1+ \nu_n \cdot (p_n^{\#}(\bm{\beta},\bm{\nu})+p_n^{\textnormal{cir}}),
 \label{eqn:phi_2v2} 
\\
  \bm{\phi}_1(\bm{\beta},\bm{\nu}) & \!:=\! [\phi_{1,n}(\bm{\beta},\bm{\nu})|_{n \in \mathcal{N}}] ,~  \bm{\phi}_2(\bm{\beta},\bm{\nu})  \!:=\! [\phi_{2,n}(\bm{\beta},\bm{\nu})|_{n \in \mathcal{N}}], \nonumber 
    \\ 
     \bm{\phi}(\bm{\beta},\bm{\nu})  &\!:=\! [\bm{\phi}_1(\bm{\beta},\bm{\nu}),\bm{\phi}_2(\bm{\beta},\bm{\nu})]. \label{eqn:phi} 
\end{talign}
With
$(\boldsymbol{p}^*, \boldsymbol{B}^*, \boldsymbol{\beta}^*)$ denoting a globally optimal solution to Problem $\mathbb{P}_{2}$ (and hence $(\boldsymbol{p}^*, \boldsymbol{B}^*)$ denoting a globally optimal solution to Problem $\mathbb{P}_{1}$), clearly setting $(\bm{\beta},\bm{\nu})$ as $(\bm{\beta}^*,\bm{\nu}^*)$ of (\ref{decidebetanstar}) and   (\ref{decidenunstar}) satisfies 
\begin{talign}
\bm{\phi}(\bm{\beta},\bm{\nu}) = \bm{0}. \label{newton}
\end{talign} Based on the above, solving Problem $\mathbb{P}_{2}$ and hence $\mathbb{P}_{1}$ can be transformed into solving~(\ref{newton}) to obtain $\mathbb{P}_3(\bm{\beta}^*,\bm{\nu}^*)$, and then setting $[\boldsymbol{p}^*, \boldsymbol{B}^*]$ as $[\bm{p}^{\#}(\bm{\beta}^*,\bm{\nu}^*),$\hspace{0pt}$ \bm{B}^{\#}(\bm{\beta}^*,\bm{\nu}^*)]$, a globally optimal solution to $\mathbb{P}_3(\bm{\beta}^*,\bm{\nu}^*)$, according to Lemma~\ref{lemma:sum-of-ratios_lemma}. Based on the above idea, we present Algorithm~\ref{algo:MN} next, where it will become clear that
\begin{talign}
&
\text{solving $\mathbb{P}_{1}$ becomes solving a series of parametric convex optimi-} \nonumber\\ & \text{zation $\mathbb{P}_3(\bm{\beta}^{(i)},\bm{\nu}^{(i)})$, with $i$ denoting the iteration index.} 
 \label{eqseries}
\end{talign}

Readers may notice that our Lemma~\ref{lemma:sum-of-ratios_lemma} provides just a necessary condition for a global optimum of Problem $\mathbb{P}_{2}$. Lemma~\ref{lemma:sum-of-ratios_lemmastronger} below shows ``necessary'' and ``sufficient'' for strictly concave utility, which holds for all types of functions in simulations of Section \ref{sec:experiments}.
\begin{lemma} \label{lemma:sum-of-ratios_lemmastronger}
If the utility function $f_n(\cdot)$ for any $n \in \mathcal{N}$ is strictly concave (i.e., $f_n'(\cdot)$ is  decreasing) for $x > 0$, the ``$\Leftarrow$'' in  (\ref{lemallresults}) can be replaced by ``$\Leftrightarrow$'', so that ``\ding{193}\,\ding{194}\,\ding{195}''\,$\Leftrightarrow$\,``\ding{192}'' actually holds in Lemma~\ref{lemma:sum-of-ratios_lemma}.  
\end{lemma}

\begin{proof}
From Theorem~\ref{thmsolveP3} on Page~\pageref{thmsolveP3}, for decreasing $f_n'(\cdot)$, we can prove that $\mathbb{P}_3\hspace{-1pt}(\bm{\beta},\hspace{-2pt}\bm{\nu})$ has a \textbf{unique} global optimum $[\bm{p}^{\#}\hspace{-1pt}(\bm{\beta},\hspace{-2pt}\bm{\nu}),\hspace{-1pt} \bm{B}^{\#}\hspace{-1pt}(\bm{\beta},\hspace{-2pt}\bm{\nu})]$. We further obtain that $(\bm{\beta}^*,\bm{\nu}^*)$ satisfying~(\ref{newton}) (i.e., ``\ding{193}\,\ding{194}\,\ding{195}'') is unique. Since we have explained that $\mathbb{P}_{2}$ has at least one global maximum, we know from the above this maximum is unique. Thus, for strictly concave utility, ``\ding{193}\,\ding{194}\,\ding{195}''\,$\Leftrightarrow$\,``\ding{192}''   holds in Lemma~\ref{lemma:sum-of-ratios_lemma}.
\end{proof}

\begin{algorithm}[!htbp]
\caption{Our approach of computing a \textbf{globally \mbox{optimal}} solution $[\bm{p}, \bm{B}]$  (up to arbitrary accuracy) to Problem $\mathbb{P}_1$ of Section~\ref{optproblem} on weighted sum-UEE optimization.}
\label{algo:MN}
Initialize feasible $[\bm{p}^{(0)}$, $\bm{B}^{(0)}]$, $i = 0$, $\xi \in (0, 1)$, $\epsilon \in (0, 1)$. \\
Calculate $\bm{\beta}^{(0)}=[\beta_n^{(0)}|_{n\in \mathcal{N}}]$ and $\bm{\nu}^{(0)}=[\nu_n^{(0)}|_{n\in \mathcal{N}}]$ via
$\beta_n^{(0)} = \frac{c_nf_n(r_{n,\textnormal{s}}(p_n^{(0)},B_n^{(0)}))}{{p_n^{(0)}+p_n^{\textnormal{cir}}}} $ and $\nu_n^{(0)} = \frac{1}{p_n^{(0)}+p_n^{\textnormal{cir}}}.$\newline \textit{//Comment: Since we aim to find $(\bm{\beta}^*,\bm{\nu}^*)$ satisfying~(\ref{newton}) (i.e., ``\ding{193}\,\ding{194}\,\ding{195}'' in Lemma~\ref{lemma:sum-of-ratios_lemma}), the above initialization is intuitively \mbox{good since it mimics (\ref{decidebetanstar}) in ``\ding{193}'' and   (\ref{decidenunstar}) in ``\ding{195}'' of~Lemma~\ref{lemma:sum-of-ratios_lemma}.}}

\Repeat{$\bm{\phi}(\bm{\beta}^{(i)},\bm{\nu}^{(i}))$ is close to $\bm{0}$}{
Use Eq.~(\ref{eqoptimal}) in Theorem~\ref{thmsolveP3} on Page~\pageref{eqoptimal} to solve $\mathbb{P}_3(\boldsymbol{\beta}^{(i)}, \boldsymbol{\nu}^{(i)})$, and obtain a solution $[\bm{p}^{\#}(\bm{\beta}^{(i)},\bm{\nu}^{(i)}), \bm{B}^{\#}(\bm{\beta}^{(i)},\bm{\nu}^{(i)})]$. \label{solveP3code}
\textit{//Comment: This line can use the bisection method in a straightforward manner, so we put the details in Appendix~\ref{applambdastar}.}

Use $[\bm{p}^{\#}(\bm{\beta}^{(i)},\bm{\nu}^{(i)}), \bm{B}^{\#}(\bm{\beta}^{(i)},\bm{\nu}^{(i)})]$ obtained above to compute $\bm{\phi}(\bm{\beta}^{(i)}, \bm{\nu}^{(i)})$ according to Eq.~(\ref{eqn:phi}) on Page~\pageref{eqn:phi}.

If $\bm{\phi}(\bm{\beta}^{(i)}, \bm{\nu}^{(i)})$ is the zero vector, then $[\bm{p}^{\#}(\bm{\beta}^{(i)},\bm{\nu}^{(i)}), \bm{B}^{\#}(\bm{\beta}^{(i)},\bm{\nu}^{(i)})]$ is the global optimal solution to Problem $\mathbb{P}_1$ and we finish the algorithm. 

Otherwise,  
 let $J_i$ be the smallest integer that satisfies\vspace{-3pt} 
\begin{talign} \label{newton_method}
    \hspace{-30pt}  & \hspace{-30pt} \| \bm{\phi}(\bm{\beta}^{(i)}+\xi^{J_i}\bm{\sigma_1}^{(i)}, \bm{\nu}^{(i)}+\xi^{J_i}\bm{\sigma_2}^{(i)}) \|_2 \notag \\
    \hspace{-30pt}  & \hspace{-30pt}\le (1-\xi^{J_i}\epsilon) \cdot \| \bm{\phi}(\bm{\beta}^{(i)}, \bm{\nu}^{(i)}) \|_2 ,\hspace{-20pt}
\end{talign}
where ``$\|\cdot\|_2$'' denotes the Euclidean norm, and the $n$th-dimension of $\bm{\sigma_1}^{(i)}$ (resp.~$\bm{\sigma_2}^{(i)}$) for $n\in \mathcal{N}$, denoted by $\bm{\sigma_1}^{(i)}[n]$ (resp.~$\bm{\sigma_2}^{(i)}[n]$), is given by\vspace{-3pt} 
\begin{talign}
  \hspace{-30pt}  & \hspace{-30pt}\textstyle{\bm{\sigma_1}^{(i)}[n]:= -\frac{ (\partial \phi_{1,n}(\bm{\beta}^{(i)},\bm{\nu}^{(i)}))/(\partial \beta_n)}{\phi_{1,n}(\bm{\beta}^{(i)},\bm{\nu}^{(i)})}} \hspace{-100pt}~\label{defsigma1} \end{talign}
~\vspace{-10pt}
\begin{talign} \hspace{-8pt}& = - \frac{\text{(\ref{bothequal})}}{\text{RHS of~(\ref{eqn:phi_1v2}) with $(\bm{\beta},\bm{\nu})$ being $(\bm{\beta}^{(i)},\bm{\nu}^{(i)})$}} , \nonumber \end{talign}
~\vspace{-10pt}
\begin{talign}
 \hspace{-30pt}  & \hspace{-30pt}\bm{\sigma_2}^{(i)}[n]:= -\frac{(\partial \phi_{2,n}(\bm{\beta}^{(i)},\bm{\nu}^{(i)}))/(\partial \nu_n)}{\phi_{2,n}(\bm{\beta}^{(i)},\bm{\nu}^{(i)})} \hspace{-100pt}~\label{defsigma2} 
\end{talign}
~\vspace{-10pt}
\begin{talign}
\hspace{-8pt}&=- \frac{\text{(\ref{bothequal})}}{\text{RHS of~(\ref{eqn:phi_2v2}) with $(\bm{\beta},\bm{\nu})$ being $(\bm{\beta}^{(i)},\bm{\nu}^{(i)})$}}, \nonumber 
\end{talign}
where RHS is short for the right-hand side.
\newline \textit{//Comment: Obtaining ${J_i}$ above involves evaluating $({J_i}+1)$ number of $\bm{\phi}(\bm{\beta}, \bm{\nu})$ for~(\ref{newton_method}). To compute each of them, we need to solve Problem $\mathbb{P}_3(\bm{\beta},\bm{\nu})$ via~(\ref{eqoptimal}) on Page~\pageref{eqoptimal} to obtain~\mbox{$[\bm{p}^{\#}\hspace{-1pt}(\bm{\beta},\hspace{-1pt}\bm{\nu}),\hspace{-1pt} \bm{B}^{\#}\hspace{-1pt}(\bm{\beta},\hspace{-1pt}\bm{\nu})]$}, and then use~(\ref{eqn:phi}).}

Update\vspace{-3pt} 
\begin{talign} \label{update_nu_beta}
  \hspace{-20pt}  & \hspace{-20pt}[\bm{\beta}^{(i+1)},\bm{\nu}^{(i+1)}] \leftarrow [ \bm{\beta}^{(i)} + \xi^{J_i}\bm{\sigma_1}^{(i)},\bm{\nu}^{(i)}+\xi^{{J_i}}\bm{\sigma_2}^{(i)}], 
\end{talign}
where ${J_i}$ is obtained from~(\ref{newton_method}). \label{lineupdate_nu_beta}
\newline \textit{//Comment: If ${J_i}$ happens to be $0$, then (\ref{update_nu_beta}) becomes the standard Newton method, as explained in the last paragraph on Page 13 of~\cite{jong2012efficient}. As shown by Problem 2 on Page 14 of~\cite{jong2012efficient}, the standard Newton method may fail for some initial points, so we follow~\cite{jong2012efficient} to find ${J_i}$ according to~(\ref{newton_method}) instead of always setting ${J_i}$ as $0$.}  

Let $i \leftarrow i+1$.
}

Use the current $[\bm{\beta},\bm{\nu}]$ in~(\ref{eqoptimal}) on Page~\pageref{eqoptimal} and return the obtained $[\bm{p}^{\#}(\bm{\beta},\bm{\nu}), \bm{B}^{\#}(\bm{\beta},\bm{\nu})]$ as the solution to Problem~$\mathbb{P}_1$.

\end{algorithm}

\subsection{Our Algorithm~\ref{algo:MN} to solve Problem $\mathbb{P}_{1}$} \label{secproveP1v2}

As explained in the previous subsection, we solve~(\ref{newton}) first in order to obtain a globally optimal solution to Problem $\mathbb{P}_{1}$.
Root-finding algorithms such as Newton's method can be used to solve~(\ref{newton}). Our Algorithm~\ref{algo:MN} actually uses a modified Newton method of~\cite{jong2012efficient}, which always converges to the desired solution. In contrast, the original Newton's method is sensitive to initialization (e.g., no convergence if starting at bad initialization, as shown in Section~4 of~\cite{jong2012efficient}).


Algorithm~\ref{algo:MN} starts with computing the initial $[\bm{\beta}^{(0)},\bm{\nu}^{(0)}]$ from $[\bm{p}^{(0)}$, $\bm{B}^{(0)}]$, as shown in the pseudocode.
In the $i$-th iteration of Algorithm~\ref{algo:MN} ($i$ starts from $0$), we update $[\bm{\beta}^{(i)},\bm{\nu}^{(i)}]$ to $[\bm{\beta}^{(i+1)},\bm{\nu}^{(i+1)}]$ based on 
(\ref{newton_method}) (\ref{defsigma1}) (\ref{defsigma2}) (\ref{update_nu_beta}), which essentially present the modified Newton method to solve~(\ref{newton}). The numerators in~(\ref{defsigma1})~(\ref{defsigma2}) use $ (\partial \phi_{1,n}(\bm{\beta}^{(i)},\bm{\nu}^{(i)}))/(\partial \beta_n)$ and $(\partial \phi_{2,n}(\bm{\beta}^{(i)},\bm{\nu}^{(i)}))/(\partial \nu_n)$, which are shown in Appendix~\ref{appbothequal} to be equal~to
\begin{talign}
p_n^{\#}(\bm{\beta}^{(i)},\bm{\nu}^{(i)})+p_n^{\textnormal{cir}}. \label{bothequal}
\end{talign}
The denominators in~(\ref{defsigma1}) (\ref{defsigma2}) use $\phi_{1,n}(\bm{\beta}^{(i)},\bm{\nu}^{(i)})$ and $\phi_{2,n}(\bm{\beta}^{(i)},\bm{\nu}^{(i)})$, whose computations based on~(\ref{eqn:phi_1v2})~(\ref{eqn:phi_2v2}) require obtaining\\ $[\bm{p}^{\#}(\bm{\beta}^{(i)},\bm{\nu}^{(i)}), \bm{B}^{\#}(\bm{\beta}^{(i)},\bm{\nu}^{(i)})]$ by solving Problem $\mathbb{P}_3(\boldsymbol{\beta}^{(i)}, \boldsymbol{\nu}^{(i)})$. This is the reason why we have~(\ref{eqseries}).

We remark that in Algorithm~\ref{algo:MN}, the iterative process of computing $[\bm{p}^{\#}(\bm{\beta}^{(i)},\bm{\nu}^{(i)}), \bm{B}^{\#}(\bm{\beta}^{(i)},\bm{\nu}^{(i)})]$ and then using it for updating $[\bm{\beta}^{(i)},\bm{\nu}^{(i)})]$ to $[\bm{\beta}^{(i+1)},\bm{\nu}^{(i+1)})]$ is not the classical dual gradient descent (DGD)~\cite{boyd2004convex} despite the resemblance, since $\bm{\beta}$ is not a Lagrange multiplier. Algorithm~\ref{algo:MN} solves~(\ref{newton}) using the modified Newton method, while DGD involves maximizing the dual function.







We formally state the solution quality of Algorithm~\ref{algo:MN} as follows. 


\begin{theoremx} \label{thm:main}

Under Conditions~\ref{conditionrn} and~\ref{fconcave} of Section~\ref{sec:prob_formu}, our proposed Algorithm~\ref{algo:MN} finds a \textbf{globally optimal} solution to Problem $\mathbb{P}_{1}$  (up to arbitrary accuracy). 
\end{theoremx}

\begin{proof}
The analyses above in Sections~\ref{secTransforming} and~\ref{secproveP1v2}, stated before Theorem~\ref{thm:main}, have already provided the proof of Theorem~\ref{thm:main}.
\end{proof}




Next, we discuss the fast convergence and order-optimal time complexity of Algorithm~\ref{algo:MN}. As shown in Theorem~3.2 of~\cite{jong2012efficient}, the modified Newton method used in Algorithm~\ref{algo:MN} has global linear and local quadratic rates of convergence.



To analyze the time complexity, we use floating point operations (flops). One addition/subtraction/multiplication/division is one flop.  We now analyze Lines 3--\ref{lineupdate_nu_beta}, the main part of Algorithm~\ref{algo:MN}. Suppose that in Line~\ref{solveP3code}, we use the bisection method to obtain $\lambda^{\#}$ from~(\ref{lambdastar}), for which there are $K$ iterations 
 and each iteration has $\mathcal{O}(N)$, where $K$ depends on the error tolerance, as detailed in Appendix~\ref{applambdastar}.
Hence, Line~\ref{solveP3code} consumes $\mathcal{O}(KN)$.
Lines 5, 6, and \ref{lineupdate_nu_beta} cost $\mathcal{O}(N)$
 flops. Line 7 takes $\mathcal{O}((J_i+1)N)$ flops. Suppose the loop in Line~3 needs $\mathcal{I}$ iterations before convergence ($\mathcal{I}$ is less than $10$ in our experiments to find a $0.01$-global optimum, which means the relative difference between the objective-function values under the found solution and the true global optimum is at most $0.01$). Then the time complexity of Algorithm~\ref{algo:MN} is $\mathcal{O}(\mathcal{I} K N + \sum_{i=0}^{\mathcal{I}-1} (J_i+1)N)$, which is linear in $N$. This linear complexity is the best that any algorithm can do, since we need to decide $N$ number of $[B_n,p_n]$ for all $N$ users. Hence, Algorithm~\ref{algo:MN} achieves the optimal time complexity in the order sense.


   

\subsection[Solving Problem P3]{Solving Problem $\mathbb{P}_3(\bm{\beta},\bm{\nu})$} \label{secsolveP3}


From~(\ref{eqseries}), solving Problem $\mathbb{P}_{1}$ requires solving a series of $\mathbb{P}_3(\bm{\beta},\bm{\nu})$. One approach is to use the Stanford CVX tool~\cite{boyd2004convex}. However, the worst-case complexity of global convex optimization grows exponentially with the problem size $N$ from Section~1.4.2 of~\cite{boyd2004convex}. 
Based on Theorem~\ref{thmsolveP3} below, we can solve $\mathbb{P}_3(\bm{\beta},\bm{\nu})$ and hence $\mathbb{P}_{1}$ in linear time with respect to~$N$, as discussed in the previous subsection.


\begin{theoremx} \label{thmsolveP3}
Under Conditions~\ref{conditionrn} and~\ref{fconcave} of Section~\ref{sec:prob_formu},
any \textbf{globally optimal} solution $[\bm{p}^{\#}(\bm{\beta},\bm{\nu}), \bm{B}^{\#}(\bm{\beta},\bm{\nu})]$ to Problem $\mathbb{P}_3(\bm{\beta},\bm{\nu})$ defined in (\ref{problem:3})  can be given as follows:
\begin{talign}
    \begin{cases}
        B_n^{\#}(\bm{\beta},\bm{\nu}) = \mathcal{B}_n(\lambda^{\#})\text{ for all $n \in \mathcal{N}$,} \\
        p_n^{\#}(\bm{\beta}, \bm{\nu}) = \frac{\sigma_n^2 B_n^{\#}(\bm{\beta},\bm{\nu}) \cdot \psi_n(\lambda^{\#})}{g_n}~\text{for all $n \in \mathcal{N}$},
    \end{cases} \label{eqoptimal}
\end{talign}
with function $\mathcal{B}_n(\lambda)$ defined by
\begin{talign}
    \mathcal{B}_n(\lambda):=\frac{\max\{ \gamma_n(\lambda),r_{n}^{\min}\}}  {\log_2\big(1+\psi_n(\lambda))},
\end{talign}
and $\lambda^{\#}$ denoting the solution to  
\begin{talign}\label{lambdastar}
 \sum\limits_{n \in \mathcal{N}} \mathcal{B}_n(\lambda) = B_{\text{total}},
\end{talign} 
where $\psi_n(\lambda)$ and $\gamma_n(\lambda)$ are defined by
\begin{talign}
& \psi_n(\lambda) := \exp\big\{1+W\big(\frac{1}{e}(\frac{g_n\lambda}{\nu_n\beta_n\sigma_n^2}-1)\big)\big\}-1, \\
&\text{for $W(\cdot)$ being the principal branch of the Lambert W function}\nonumber \\
&\text{($W(z)$ for $z \geq -e^{-1}$ is the solution of $x \geq -1$ to the equation $x e^{x} = z$)}, \nonumber  \\
&\text{and } \gamma_n(\lambda) - r_{n, \textnormal{e}}:= \begin{cases}
\xi:=(f_n^{\prime})^{-1}\Big(\frac{\beta_n\sigma_n^2\cdot(1+\psi_n(\lambda))\ln{2}}{c_ng_n}\Big) \text{ }\\ ~~~~~\text{when such result $\xi \geq 0$ exists},\\
0,~\text{otherwise},
\end{cases} 
\end{talign}
with $(f_n^{\prime})^{-1}(\cdot)$ denoting the inverse function of the derivative $f_n^{\prime}(\cdot)$.



\end{theoremx}
Theorem~\ref{thmsolveP3} is proved 
in Appendix~\ref{Appendix:theorem6.3}, where we also elaborate on the bisection method to obtain $\lambda^{\#}$ from~(\ref{lambdastar}).
\section{Broad Usage of Our Technique} \label{Insights}

In this section, we review the optimization used in Algorithm~\ref{algo:MN} to obtain an insightful technique, which can be used to solve many other problems in wireless networks and mobile computing.



In Section~\ref{secTransforming}, Problem $\mathbb{P}_2$ is not convex optimization since the \mbox{non-convex} product \mbox{$\beta_{n}\cdot(p_n + p_n^{\textnormal{cir}})$} exists in~(\ref{constra:pysical_epi_data_rate}), as shown in the sentences following (\ref{defineFnpnBn}). The solving process of $\mathbb{P}_2$ is transformed into solving a series of parametric convex optimization $\mathbb{P}_3(\bm{\beta},\bm{\nu})$ where $[\bm{\beta},\bm{\nu}]$ is given so that there is no \mbox{non-convex} product term and we have convex optimization. The solving of each $\mathbb{P}_3$ is used to update $[\bm{\beta},\bm{\nu}]$ under which $\mathbb{P}_3$ is solved again with the new  
$[\bm{\beta},\bm{\nu}]$, where the update of $[\bm{\beta},\bm{\nu}]$ is based on the KKT conditions of $\mathbb{P}_2$. 

From the above discussion, we can identify the following:

\textbf{Our technique to handle functions of product or quotient terms in optimization:} \textit{With }``$\divideontimes$'' \textit{denoting multiplication or division, if there are terms $f_n(A_n(\boldsymbol{x}) \divideontimes y_n)|_{n\in\mathcal{N}}$ in an optimization problem $\mathbb{P}$, for functions $f_n,A_n|_{n\in\mathcal{N}}$ and variables $\boldsymbol{x}$ and $\boldsymbol{y}=[y_n|_{n\in\mathcal{N}}]$, we can convert $\mathbb{P}$ into a series of parametric convex optimization $\mathbb{Q}(\boldsymbol{y},\boldsymbol{z})$, where $\boldsymbol{z}$ comprises additional variables in the parameterization (e.g., $\bm{\nu}$ in our }``$\mathbb{P}_3(\bm{\beta},\bm{\nu})$''\textit{). In $\mathbb{Q}(\boldsymbol{y},\boldsymbol{z})$, given $[\boldsymbol{y},\boldsymbol{z}]$, variables in $A_n(\boldsymbol{x}) \divideontimes y_n$ just have $\boldsymbol{x}$, so that $\mathbb{Q}$ can be easier to solve than $\mathbb{P}$, or $\mathbb{Q}$ may even happen to be convex in $\boldsymbol{x}$. The solving of each $\mathbb{Q}$ will be used to update $[\boldsymbol{y},\boldsymbol{z}]$ under which $\mathbb{Q}$ is solved again with the new  
$[\boldsymbol{y},\boldsymbol{z}]$, where the update of $[\boldsymbol{y},\boldsymbol{z}]$ is based on the KKT conditions of~$\mathbb{P}$.} 

With the above technique, we can  address $f_n(A_n(\boldsymbol{x}) \divideontimes B_n(\boldsymbol{x}))|_{n\in\mathcal{N}}$ in optimization as well, for functions $f_n,A_n,B_n|_{n\in\mathcal{N}}$ and variables $\boldsymbol{x}$. We  replace $A_n(\boldsymbol{x}) \divideontimes B_n(\boldsymbol{x})$ by an auxiliary variable $z_n$ and enforce the constraint of $z_n$ being either no greater or no less than $A_n(\boldsymbol{x}) \divideontimes  B_n(\boldsymbol{x})$ (depending on the specific problem), where the constraint can be further converted into a relationship between $A_n(\boldsymbol{x})$ and $z_n \divideontimes B_n(\boldsymbol{x})$, like how we transform $\mathbb{P}_1$ of~(\ref{problem:1}) to $\mathbb{P}_2$ of~(\ref{problem:2}).

To summarize, our technique can be useful for various optimization problems involving product or quotient terms. In addition, the technique often obtains a global optimum, as in Theorem~\ref{thm:main}. The above finding goes beyond the sum-of-ratios optimization of~\cite{jong2012efficient}, although our original motivation comes from~\cite{jong2012efficient}. The following discussion shows that our above finding is very likely to be new.


 
 Two recent papers~\cite{shen2018fractional,shen2018fractional2} by Shen and Yu have been considered breakthroughs in fractional programming, as seen from their high citations (692 and 190, respectively, as of 10 March 2023 in Google Scholar). However, they find neither local nor global optimum. In contrast, our technique above will find a global optimum. Interested readers can refer to
Appendix~\ref{secfractionalprogramming}.





Our above technique can be applied to many optimization problems in wireless networks and mobile computing, as illustrated by two examples below. In interference-constrained wireless networks, globally solving  the weighted sum-rate maximization (WSRM) efficiently was an open problem for years before it was addressed by~\cite{qian2009mapel}, since a user's rate (per unit bandwidth) given by  $\log_2(1+\frac{\text{TransmitPower}}{\text{Interference}+\text{Noise}})$ involves a fraction inside a logarithm, which is difficult to deal with. Our technique above will find a global optimum for WSRM and other problems involving the above rate expression, while the polyblock-based approach of~\cite{qian2009mapel} relies on the structure of WSRM and may not be applicable to other problems. In mobile edge computing, with $\gamma$ denoting the offloading ratio of computation tasks,~\cite{zhao2021energy} minimizes the system cost, given by \mbox{$\gamma \hspace{-2pt}\cdot\hspace{-2pt} \text{EdgeComputingCost} \hspace{-2pt}+ \hspace{-2pt}(1\hspace{-2pt}-\hspace{-2pt}\gamma) \hspace{-2pt}\cdot\hspace{-2pt} \text{LocalComputingCost}$}. The multiplication above means no joint convexity in $\gamma $ and other variables. Then~\cite{zhao2021energy} uses alternating optimization which is neither locally nor globally optimal, while our technique will find a global optimum.

\section{Simulation}\label{sec:experiments}

The utility functions for simulations are presented in
Section~\ref{secTypicalUtility}  and validated by real data in
Section~\ref{realdata}.
Then 
we describe simulation settings in Section \ref{experiment:parameter}, before reporting results in other subsections. 

\subsection{Utility functions for simulation} \label{secTypicalUtility}



We provide three types of utility functions below since the Metaverse offers diverse applications. In Section~\ref{realdata},
 we  validate these functions using real data. 

%

\textbf{Type 1 utility function: } We have
\begin{align}
f_n(x)=\kappa_n \ln (b_n+ a_n x), \label{deftype1}
\end{align}
where $a_n, \kappa_n >0, b_n \geq 0$.   This type is used in \cite{yang2012crowdsourcing} for sensing tasks. In simulations starting from Section~\ref{experiment:parameter}, we let $b_n=1$.



\textbf{Type 2 utility function: } We have
\begin{align}
    f_n(x)= \kappa_n \cdot (1- e^{-a_n x+c_n}), \label{deftype2}
\end{align}
where $a_n, \kappa_n >0$, and $e$ denotes Euler's number. This type is motivated by \cite{liu2018edge} on augmented reality. We let $c_n=0$ in simulations.


\textbf{Type 3 utility function: } We have
\begin{align}
    f_n(x)=\kappa_n ({x}+d_n)^{a_n}, \label{deftype3}
\end{align}
where $\kappa_n >0$, $d_n \geq 0$, and $0<a_n<1$. This function form has been used in prior work on congestion control~\cite{mo2000fair} and mobile data subsidization~\cite{xiong2020reward}.  We let $d_n=0$ in simulations. In the terminologies of economics, $\kappa_n {r_{n,\textnormal{s}}}^{a_n}$ can be viewed as a Cobb--Douglas utility with respect to $r_{n,\textnormal{s}}$, while $\kappa_n \cdot {(r_n - r_{n,\textnormal{e}})}^{a_n}$ can be regarded as a Stone--Geary utility with respect to $r_n$; see Page 7 of~\cite{da2021welfare}.



It is straightforward to show that the above three types  for the utility function all satisfy Condition~\ref{fconcave} of Section~\ref{secfConditions}.  
These three types are what we will use in simulations from Section~\ref{experiment:parameter}. 
We emphasize that our theoretical results (e.g., Algorithm~\ref{algo:MN} as well as Theorems~\ref{thm:main} and~\ref{thmsolveP3} in Section~\ref{secAlgorithm}) of this paper apply to \textbf{any} utility function satisfying Condition~\ref{fconcave} of Section~\ref{secfConditions}.



\subsection{\mbox{Real data validating utility functions above}} \label{realdata}

We now validate Section~\ref{secTypicalUtility}'s utility functions with the SSV360~\cite{ssv360}  and Netflix datasets~\cite{Netflix} from real-world experiments.

\textbf{SSV360 dataset.} This dataset of~\cite{ssv360} captures users' assessment of 360\degree~videos when wearing HTC Vive Pro virtual reality headsets. Each data point represents a user's subjective quality assessment of a 360\degree~scene, under  standing or seated viewing (SSV). 
In the dataset, having data points under different video bitrates yet the same resolution is due to different quantization parameters used in video compression. The wireless data rate should be large enough to ensure a smooth watching experience at the given video bitrate~\cite{liu2018edge}. We  let  the  bitrate be a constant fraction (say $\theta$) of the wireless rate. Since changing the bitrate $r_{\text{bitrate}}$ to the wireless rate $r_{\text{wireless}}$ just involves replacing $r_{\text{bitrate}}$ with $r_{\text{wireless}}/\theta$, we perform curve-fitting with the bitrate to validate the utility functions.
The curves in Fig.~\ref{fig:real_simul}(a) are for the scenarios of ``user 1 seated'', ``user 2 seated'', and ``user 1 standing'' respectively, to watch the same 360\degree~scene ``FormationPace''~\cite{ssv360} with 2K resolution (i.e., $2048 \times 1080$ pixels).

In the SSV360 dataset, the score follows the widely used Absolute Category Rating (ACR)~\cite{gutierrez2022vqeg} and is an integer from $1$ to $5$. To obtain better curve-fitting results, we further use the Netflix dataset, where the score (i.e., the $y$-axis) ranges from $0$ to $100$.



\textbf{Netflix dataset.} In this dataset~\cite{Netflix}, which is a part of Netflix's Emmy Award-winning Video Multimethod Assessment Fusion (VMAF) project, each data point exhibits users' mean opinion score in $[0,100]$ for a video at a given resolution and a given bitrate. Because there are not enough data points that have  different bitrates yet the
same resolution, we treat both  resolution and bitrate as variables for curve-fitting. The results are shown in Fig. \ref{fig:real_simul}(b).




The expressions for the curves in Fig.~\ref{fig:real_simul} are in the table below.\vspace{4pt}



%

\setlength{\tabcolsep}{1pt}

\noindent\begin{tabular}{l|l|l}
\hline
        Dataset         & \begin{tabular}[c]{@{}l@{}}Scenario in~\cite{ssv360} \\ or video in~\cite{Netflix} \end{tabular} &  \begin{tabular}[c]{@{}l@{}} \small Utility function from curve-fitting for \\[-3pt] \small  normalized  bitrate $x$ and normalized\\[-3pt] \small  resolution $y$ explained in the caption of Fig.~\ref{fig:real_simul} \end{tabular} \\ \hline
\multirow{3}{*}{\begin{tabular}[c]{@{}l@{}}SSV360\\in~\cite{ssv360}\end{tabular}} & user 1 seated  & Type 1: $ 0.5424 \ln (1+ 37.2965 x)$ \\ \cline{2-3} 
                 & user 2 seated &  \begin{tabular}[c]{@{}l@{}} ~\\[-10pt] Type 2: $2.9351 (1- e^{-2.1224 x})$   \end{tabular} \\ \cline{2-3} 
                  &  user 1 standing & \begin{tabular}[c]{@{}l@{}}~\\[-10pt] Type 3: $3.2956 {(x/15.94)}^{0.2733}$\end{tabular}  \\ \hline
\multirow{3}{*}{\begin{tabular}[c]{@{}l@{}}Netflix\\in~\cite{Netflix}\end{tabular}}
                  &ElFuente1  & Type 1: $33.4215\ln(1+0.784x+10.0826y)$ \\ \cline{2-3} 
                  &BigBuckBunny  & \begin{tabular}[c]{@{}l@{}}~\\[-10pt] Type 2: $103.3464(1-e^{-0.23166x-2.9792y})$\end{tabular} \\ \cline{2-3}  &BirdsInCage  & \begin{tabular}[c]{@{}l@{}}~\\[-10pt] Type 3: $61.8622(x/15+y/1.1664)^{0.5301}$\end{tabular} \\ \hline
\end{tabular}

~\vspace{0pt}

The existence of $y$ in some expressions 
above can be understood that the coefficients in (\ref{deftype1}) (\ref{deftype2}) (\ref{deftype3}) depend on $y$. In simulations below, we fix $y$ so that the utility function depends on only the rate.\vspace{-2pt}


\subsection{Parameter setting}\label{experiment:parameter}

We first state settings that apply to all simulations. Based on~\cite{zhou2022resource}, we model the path loss between each legitimate user and the Metaverse server as $128.1+37.6\log(distance)$ along with 8 decibels (dB) for the standard deviation of shadow fading, and the unit of $distance$ is kilometer. The power spectral density of Gaussian noise ${\sigma_n}^2$ is $-174$\,dBm/Hz (i.e., 4\,zeptowatts/Hz, the value for thermal noise at 20\,\degree{C} room temperature~\cite{huang2013noise}).

In addition, some default settings are as follows, unless otherwise specified. $N$ denoting the number of legitimate users is 30. The weight parameter $c_n$ is set to $1$ for all users (unless configured otherwise), which means the weighted sum-UEE just becomes sum-UEE by default. 
The default total bandwidth $B_{\textnormal{total}}$ is 20\,MHz. The circuit power $p_n^{cir}$ is 2\,dBm (i.e., 1.6\,milliwatts) for each $n$. Both the eavesdropping rate $r_{n,e}$ and the minimum transmission rate $r_n^{\min}$ are $20$\,kilobits per second (Kbps) by default. For the utility functions, we set $\kappa_n=1$, $a_n=0.5$, $b_n=1$, $c_n=0$, and $d_n=0$ by default. In all simulations, we  stop the algorithm after obtaining a $0.01$-global optimum, whose meaning is discussed at the end of Section~\ref{secproveP1v2}.


\subsection{Comparison of different algorithms} \label{secComparison}

We compare our Algorithm~\ref{algo:MN} with the following baselines: 
\begin{itemize}
\item[(i)] \textbf{Optimize $\bm{B}$ only}: Here we let $p_n$ for each $n$ be $1$~milliwatt (i.e., $10^{-3}$ W), which will be substituted into Problem $\mathbb{P}_{1}$. Then ``optimizing $\bm{B}$ only'' becomes convex optimization, for which the KKT conditions are analyzed to obtain the solution. 
\item[(ii)] \textbf{Optimize $\bm{p}$ only}: In this case, we let $B_n$ for each $n$ be $B_{\text{total}}/N$, which will be substituted into Problem $\mathbb{P}_{1}$. Then ``optimizing $\bm{p}$ only'' belongs to convex optimization, for which the KKT conditions are inspected to acquire the solution.
\item[(iii)]  \textbf{Alternating optimization}: Starting with a feasible initialization, we perform ``(i)'' and ``(ii)'' above in an alternating manner, until convergence (when the relative improvement between two consecutive iterations is negligible). 
\end{itemize} 

For the detailed analyses of the baseline algorithms, interested readers can refer to
Appendix~\ref{sec:baseline model}.



\begin{figure}[!t]
    \centering
    \includegraphics[width=0.47\textwidth]{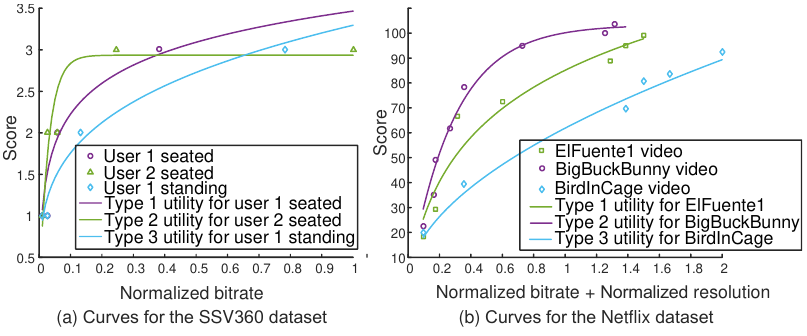}
\vspace{-12pt}    \caption{The curve-fitting results using Types 1, 2, and 3 utility functions. In both subfigures, the normalized video bitrate is obtained from dividing the bitrate by the maximum value of $15.94$\,Mbps (resp., $15$\,Mbps) in the SSV360 (resp., Netflix) dataset. In Fig.~\ref{fig:real_simul}(b), the normalized resolution is the result of dividing the resolution by the maximum resolution of $1920 \times 1080$ pixels (i.e.,  $1080$p) in the Netflix dataset.
  \vspace{-5pt} }
    \label{fig:real_simul}
\end{figure}

\begin{figure}[!t]
 \hspace{-12pt}\includegraphics[width=.47\textwidth]{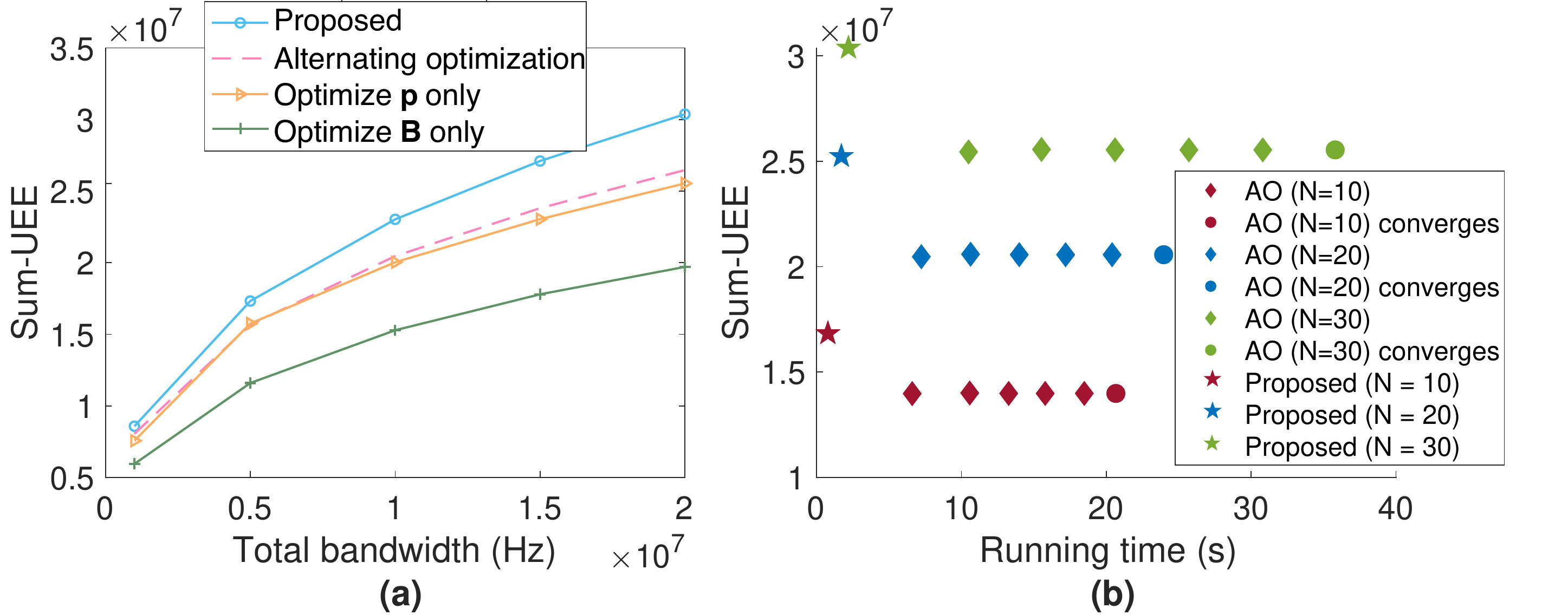}
  \vspace{-12pt}    \caption{(a). Algorithms with respect to the total bandwidth.
    (b). Running time and objective-function value under each algorithm, where AO is short for alternating optimization.\vspace{-10pt}  \vspace{-5pt} }
    \label{fig:Comparision with baseline}
\end{figure}

\begin{figure}[!t]
    \centering
        \includegraphics[width=.46\textwidth]{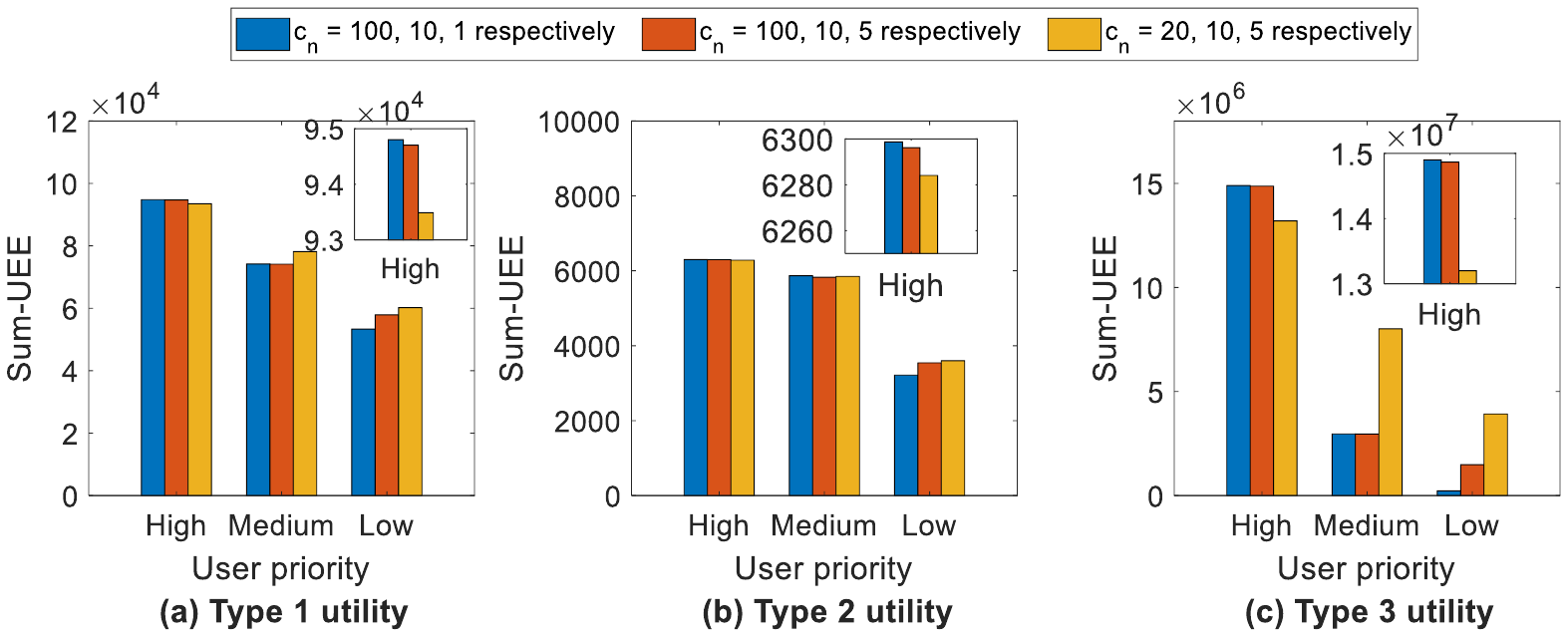}\vspace{-12pt}
        \caption{Sum-UEE for high-, medium-, and low-priority users, under different $c_n$ values, for Type 1, 2, or 3 utility. \vspace{-8pt} 
        }
        \label{fig:c_n}
\end{figure}

\begin{figure}
    \centering
    \includegraphics[width=.37\textwidth]{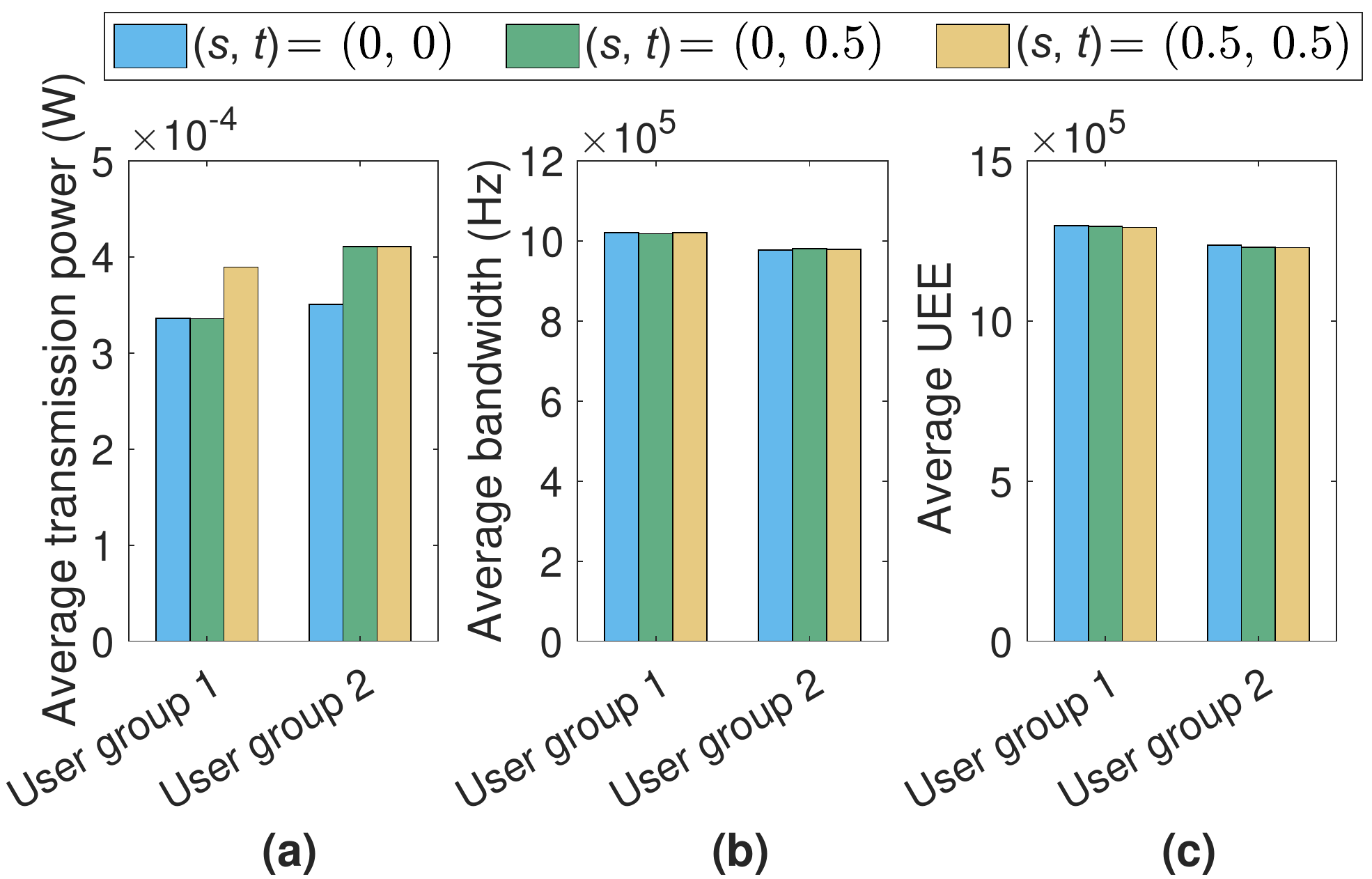}\vspace{-12pt}
    \caption{Setting different $r_{n,\textnormal{e}}$ for two user groups, where the legend ``$(s,t)$'' means that $r_{n,\textnormal{e}}$ for user $U_n$ in Group 1 is set as $s \cdot r_n^{\min}$ while $r_{n,\textnormal{e}}$ for user $U_n$ in Group 2 equals $t \cdot r_n^{\min}$.\vspace{-10pt}  } 
    \label{fig:diff_re_rmin}
\end{figure}

We compare Algorithm~\ref{algo:MN} with the above baselines in 
Fig.~\ref{fig:Comparision with baseline}, where Type 3 utility function is used.
Fig. \ref{fig:Comparision with baseline}(a) plots the sum-UEE with respect to the total bandwidth $B_{\textnormal{total}}$, and shows that our proposed algorithm achieves a larger sum-UEE than all the baselines. For each curve, as $B_{\textnormal{total}}$ increases, the sum-UEE grows but at a  slower rate. Formally proving this is not in this paper's scope, but it seems intuitive since each user's UEE is an increasing and concave function of the bandwidth, given the transmission power.

Fig. \ref{fig:Comparision with baseline}(b) displays the convergence performance of our proposed algorithm and alternating optimization (AO).
Our  algorithm always obtains a better sum-UEE and converges much faster. In particular, on a laptop with 8GB of RAM and 256GB of storage, the proposed algorithm converges within 3 seconds, but the AO approach takes around $20$, $24$, $36$ seconds under $N = 10, 20, 30$, respectively.




\subsection{The priority levels of users}

Here we explore the influence of the priority of users under different utility functions. 
We consider that $30$ users are evenly classified into three priority levels, corresponding to different weights $c_n$. Larger $c_n$ means more weight in our studied optimization.  
For example, the legend ``$c_n=100,10,1$'' in Fig.~\ref{fig:c_n} means that $10$ users with $c_n$ being $100$ (resp., $10$ and $1$) have high (resp., medium and low) priority. 

Fig.~\ref{fig:c_n}(a) (b) and (c) use utility functions of Types 1, 2, and 3, respectively. The sum-UEE of users in each priority group is plotted. In each subfigure, the bar charts show that the sum-UEE of the high-priority group is the largest, while that of the low-priority group is the lowest, matching the intuition, since higher priority means larger $c_n$ and ``more say'' in the weighted sum-UEE optimization. In addition, the  numbers in Fig.~\ref{fig:c_n}(c) for Type 3 utility $u_3:=r_{n,\textnormal{s}}^{0.5}$ are greater than the corresponding ones in Fig.~\ref{fig:c_n}(a) for Type 1 utility $u_1:=\ln(1+0.5r_{n,\textnormal{s}})$, which are further larger than those in Fig.~\ref{fig:c_n}(b) for Type 2 utility $u_2:=1-e^{-0.5r_{n,\textnormal{s}}}$. The above is consistent with  $u_3>u_1>u_2$ for large $r_{n,\textnormal{s}}$ (in the unit of bps). 

In Fig.~\ref{fig:c_n}'s subfigures, from Case~1 of ``$c_n=100, 10, 1$'' to  Case~2 of ``$c_n=100, 10, 5$'', and further to Case 3 of ``$c_n=20,10,5$'', the relative dominance of high-priority group decreases while  the relative weight of low-priority group increases, which accords with  declining (resp., rising) sum-UEE of high-priority (resp., low-priority) group from Case 1 to Case 2, and further to Case 3. For the medium-priority group, as expected, the sum-UEE decreases from Case 1 to Case 2 (though not clear in the plots without zooming in), and increases from Case 2 to Case 3. The above demonstrates the impact of the weight $c_n$ as the priority level.


\subsection{Impact of individual rate constraints}

Now we report the effect of varying $r_{n,\textnormal{e}}$. The number of users $N$ is set as 20, and we divide them equally into two groups with $r_{n,\textnormal{e}}$ as $s \cdot r_n^{\min}$ and $t \cdot r_n^{\min}$, respectively, where $r_n^{\min}$ is the default $20$\,Kbps. Each group's average transmission power, average allocated bandwidth, and average UEE are plotted in Fig. \ref{fig:diff_re_rmin}(a), (b), and (c), respectively. Also, Fig. \ref{fig:diff_re_rmin} uses Type 3 utility function and evaluates $(s,t)$ as $(0,0)$, $(0,0.5)$, and $(0.5,0.5)$, respectively. From  $(0,0)$ to $(0,0.5)$ (resp., $(0,0.5)$ to $(0.5,0.5)$), the second (resp., first) group's average transmission power increases. This is intuitive since raising a group's $r_{n,\textnormal{e}}$ with everything else unchanged requires the group to enlarge the transmission power and hence the data rate. The bandwidth allocation does not vary much under the cases of $(0,0)$, $(0,0.5)$, and $(0.5,0.5)$. For each group, the average UEE slightly drops as $r_{n,\textnormal{e}}$ grows, which seems intuitive since each user's utility $f_n(r_n(p_n,B_n)-r_{n,\textnormal{e}})$ is negatively correlated with $r_{n,\textnormal{e}}$.

For more simulation results (e.g., how our proposed algorithm performs when the number of users changes or when there are heterogeneous types of utility functions among the users), interested readers can refer to
Appendices~\ref{appImpactnumberusers} and~\ref{secHeterogeneous}.\vspace{-1pt}

\section{Conclusion\vspace{-1pt}}\label{sec:conclusion}
In this paper, in a wireless network for the Metaverse, we have studied the weighted optimization of all users'  utility-energy efficiency (UEE) under physical-layer security. The formulated problem belongs to \mbox{non-convex} optimization, and we solve it via a transform to parametric convex optimization. The resulting algorithm is optimal in terms of both the solution quality and the order of time complexity. Simulation results are provided with utility functions validated by real data. We envision more research to adopt our transform technique due to its broad applicability to other  problems in wireless networks and mobile computing.\vspace{-1pt}

\newpage

\renewcommand{\theequation}{A.\arabic{equation}}
  \setcounter{equation}{0}  
\begin{appendix}

\newpage

\noindent\textbf{\huge	 Appendices:}

~

We introduce some notation to be used in the appendices.
For a scalar function $f(x_1,x_2,\ldots,x_M)$ of $M$ variables $x_1,x_2,\ldots,x_M$, we use $\nabla_{x_m}f(x_1,x_2,\ldots,x_M)$ where $m\in\{1,2,\ldots,M\}$ to denote the partial derivative of $f(x_1,x_2,\ldots,x_M)$ with respect to $x_m$, and use $\nabla_{x_m}f(x_1,x_2,\ldots,x_M)|_{x_m = x_m^*}$ to denote the corresponding result when $x_m$ equals a given value $x_m^*$. For a $K$-element set $\{x_{i_1}, x_{i_2},\ldots,x_{i_K}\}$ of variables, which is a subset of $\{x_1,x_2,\ldots,x_M\}$, we define \\ $\nabla_{x_{i_1}, x_{i_2},\ldots,x_{i_K}}f(x_1,x_2,\ldots,x_M)$
as the vector \\
$[\nabla_{x_{i_k}}f(x_1,x_2,\ldots,x_M) |_{k=1,2,\ldots,K}]$.

\section{Explaining~(\ref{bothequal})} \label{appbothequal}

To establish~(\ref{bothequal}), we will prove
\begin{talign}
&\text{$ (\partial \phi_{1,n}(\bm{\beta},\bm{\nu}))/(\partial \beta_n)$ and $(\partial \phi_{2,n}(\bm{\beta},\bm{\nu}))/(\partial \nu_n)$ equal }
p_n^{\#}(\bm{\beta},\bm{\nu})+p_n^{\textnormal{cir}}. \label{bothequalv2}
\end{talign}
The definitions of $\phi_{1,n}(\bm{\beta},\bm{\nu})$ and $\phi_{2,n}(\bm{\beta},\bm{\nu})$ in~(\ref{eqn:phi_1v2}) and~(\ref{eqn:phi_2v2}) use $[\bm{p}^{\#}(\bm{\beta},\bm{\nu}), \bm{B}^{\#}(\bm{\beta},\bm{\nu})]$, which denotes
a globally optimal solution  to  $\mathbb{P}_3(\bm{\beta},\bm{\nu})$. Hence, below we analyze $\mathbb{P}_3(\bm{\beta},\bm{\nu})$.




Problem $\mathbb{P}_3(\bm{\beta},\bm{\nu})$ belongs to convex optimization and Slater's condition holds, as shown in the proof of Lemma~\ref{lemma:sum-of-ratios_lemma}. Then the Karush--Kuhn--Tucker (KKT) conditions are necessary and sufficient to obtain the globally optimal solution, as stated in Lemma~\ref{lemmacvxuseKKT}. To this end, we define the Lagrange function:
 \begin{align}
   & \textstyle{  L_{\mathbb{P}_3}(\boldsymbol{p}, \boldsymbol{B}, \boldsymbol{\tau}, \lambda \mid \boldsymbol{\beta},\boldsymbol{\nu})  
     = -\sum_{n\in \mathcal{N}} \mathcal{F}_n(p_n, B_n|\beta_n, \nu_n)} \nonumber
    \\
     & \textstyle{\quad +  \sum_{n\in \mathcal{N}} \tau_n\cdot(r_n^{\min} -r_n ) + \lambda\cdot(\sum_{n\in \mathcal{N}}B_n- B_{\text{total}})}, \label{eqLP3}
 \end{align}
 where $\bm{\tau}$ and $\lambda$ are called Lagrange multipliers.

The KKT conditions of Problem $\mathbb{P}_3(\bm{\beta},\bm{\nu})$ are as follows, with $L_{\mathbb{P}_3}$ short for $L_{\mathbb{P}_3}(\boldsymbol{p}, \boldsymbol{B}, \boldsymbol{\tau}, \lambda \mid \boldsymbol{\beta},\boldsymbol{\nu})  $:
\begin{subequations} \label{P3KKTeq}
\begin{talign} 
&\textnormal{\textbf{Stationarity:}} \nonumber
\\  &\frac{\partial L_{\mathbb{P}_3}}{\partial p_n}     = 0, \text{ for all } n \in \mathcal{N}, \label{P3KKTeqStationaritypn} \\ &\frac{\partial L_{\mathbb{P}_3}}{\partial B_n}     = 0, \text{ for all } n \in \mathcal{N}, \label{P3KKTeqStationarityBn}   \\ &\textnormal{\textbf{Complementary slackness:}} \nonumber  \\ & \tau_n\cdot(r_n^{\min} -r_n(p_n,B_n)) = 0, \text{ for all } n \in \mathcal{N}, \label{P3KKTeqComplementarytaun}  \\ &  \lambda\cdot(\sum_{n\in \mathcal{N}}B_n- B_{\text{total}}) =0; \label{P3KKTeqComplementarylambda}  \\ &\textnormal{\textbf{Primal feasibility:}} \nonumber
\\  &  r_n(p_n,B_n) \geq r_n^{\min}, \text{for all } n \in \mathcal{N}, \label{P3KKTeqPrimalfeasibility2}\\  & \sum_{n\in \mathcal{N}}B_n \leq B_{\text{total}}; \label{P3KKTeqPrimalfeasibility3}\\ &\textnormal{\textbf{Dual feasibility:}} \nonumber \\ &  \tau_n  \geq 0,  \text{for all } n \in \mathcal{N}. \label{P3KKTeqDualfeasibilitytau}\\ &  \lambda  \geq 0. \label{P3KKTeqDualfeasibility}  
\end{talign}
\end{subequations}

Recall that we use $[\bm{p}^{\#}(\bm{\beta},\bm{\nu}), \bm{B}^{\#}(\bm{\beta},\bm{\nu})]$ to denote
a globally optimal solution  to  $\mathbb{P}_3(\bm{\beta},\bm{\nu})$, where $\bm{p}^{\#}(\bm{\beta},\bm{\nu})=[p_n^{\#}(\bm{\beta},\bm{\nu})|_{n \in \mathcal{N}}]$ and $\bm{B}^{\#}(\bm{\beta},\bm{\nu})=[B_n^{\#}(\bm{\beta},\bm{\nu})|_{n \in \mathcal{N}}]$. Hence, $[\bm{p}^{\#}(\bm{\beta},\bm{\nu}), \bm{B}^{\#}(\bm{\beta},\bm{\nu})]$ satisfies the KKT conditions above. Using~(\ref{eqLP3})~(\ref{P3KKTeqComplementarytaun}) and~(\ref{P3KKTeqComplementarylambda}), we know
 \begin{talign}
&\textstyle{  L_{\mathbb{P}_3}(\bm{p}^{\#}(\bm{\beta},\bm{\nu}), \bm{B}^{\#}(\bm{\beta},\bm{\nu}), \boldsymbol{\tau}, \lambda \mid \boldsymbol{\beta},\boldsymbol{\nu})}  
   \nonumber
    \\
     &   = -\sum_{n\in \mathcal{N}} \mathcal{F}_n(p_n^{\#}(\bm{\beta},\bm{\nu}), B_n^{\#}(\bm{\beta},\bm{\nu})|\beta_n, \nu_n) \nonumber
    \\
     & \textstyle{\quad +  \sum_{n\in \mathcal{N}} \tau_n\cdot(r_n^{\min} -r_n(p_n^{\#}(\bm{\beta},\bm{\nu}), B_n^{\#}(\bm{\beta},\bm{\nu})) )}\nonumber
    \\
     & \quad + \lambda\cdot(\sum_{n\in \mathcal{N}}B_n^{\#}(\bm{\beta},\bm{\nu})- B_{\text{total}})\nonumber
    \\
     & =-\sum_{n\in \mathcal{N}} \mathcal{F}_n(p_n^{\#}(\bm{\beta},\bm{\nu}), B_n^{\#}(\bm{\beta},\bm{\nu})|\beta_n, \nu_n). \label{eqlp3f}
 \end{talign}

For notation simplicity, we group $\bm{\beta} $ and $\bm{\nu}$ together to define $\bm{\alpha} $; i.e., $\bm{\alpha}:=[\bm{\beta},\bm{\nu}]$. Note that~(\ref{eqlp3f}) above holds for any $\bm{\alpha} $. Hence, it holds that
\begin{talign}
&\nabla_{\bm{\alpha}} \left(\sum_{n\in \mathcal{N}} \mathcal{F}_n(p_n^{\#}(\bm{\beta},\bm{\nu}), B_n^{\#}(\bm{\beta},\bm{\nu})|\beta_n, \nu_n)\right)  \nonumber
    \\
     &  = -\nabla_{\bm{\alpha}} \left(L_{\mathbb{P}_3}(\bm{p}^{\#}(\bm{\beta},\bm{\nu}), \bm{B}^{\#}(\bm{\beta},\bm{\nu}), \boldsymbol{\tau}, \lambda \mid \boldsymbol{\beta},\boldsymbol{\nu})\right).\label{eqlp3fderi}
\end{talign}
Based on (\ref{eqlp3fderi}), to get $\nabla_{\bm{\alpha}} \left(\sum_{n\in \mathcal{N}} \mathcal{F}_n(p_n^{\#}(\bm{\beta},\bm{\nu}), B_n^{\#}(\bm{\beta},\bm{\nu})|\beta_n, \nu_n)\right)$, we  compute $\nabla_{\bm{\alpha}} \left(L_{\mathbb{P}_3}(\bm{p}^{\#}(\bm{\beta},\bm{\nu}), \bm{B}^{\#}(\bm{\beta},\bm{\nu}), \boldsymbol{\tau}, \lambda \mid \boldsymbol{\beta},\boldsymbol{\nu})\right)$ and take its negative. We have 
\begin{talign}
& \nabla_{\bm{\alpha}} \left(L_{\mathbb{P}_3}(\bm{p}^{\#}(\bm{\beta},\bm{\nu}), \bm{B}^{\#}(\bm{\beta},\bm{\nu}), \boldsymbol{\tau}, \lambda \mid \boldsymbol{\beta},\boldsymbol{\nu})\right)\nonumber \\ & =\sum_{n\in \mathcal{N}} \left\{ \left[(\nabla_{p_n} L_{\mathbb{P}_3}(\boldsymbol{p}, \boldsymbol{B}, \boldsymbol{\tau}, \lambda \mid \boldsymbol{\beta},\boldsymbol{\nu}))|_{\begin{subarray}{l}\bm{p}=\bm{p}^{\#}(\bm{\beta},\bm{\nu}), \\ \bm{b}=\bm{b}^{\#}(\bm{\beta},\bm{\nu})\end{subarray}}\right] \cdot \nabla_{\bm{\alpha}} p_n^{\#}(\bm{\beta},\bm{\nu})\right\}\nonumber \\ & \quad +\sum_{n\in \mathcal{N}} \left\{ \left[(\nabla_{b_n} L_{\mathbb{P}_3}(\boldsymbol{p}, \boldsymbol{B}, \boldsymbol{\tau}, \lambda \mid \boldsymbol{\beta},\boldsymbol{\nu}))|_{\begin{subarray}{l}\bm{p}=\bm{p}^{\#}(\bm{\beta},\bm{\nu}), \\ \bm{b}=\bm{b}^{\#}(\bm{\beta},\bm{\nu})\end{subarray}}\right] \cdot \nabla_{\bm{\alpha}} b_n^{\#}(\bm{\beta},\bm{\nu})\right\}  \nonumber \\ & \quad + \left[(\nabla_{\bm{\alpha}} L_{\mathbb{P}_3}(\boldsymbol{p}, \boldsymbol{B}, \boldsymbol{\tau}, \lambda \mid \boldsymbol{\beta},\boldsymbol{\nu}))|_{\begin{subarray}{l}\bm{p}=\bm{p}^{\#}(\bm{\beta},\bm{\nu}), \\ \bm{b}=\bm{b}^{\#}(\bm{\beta},\bm{\nu})\end{subarray}}\right].\label{eqlp3fderiv2}
\end{talign}
From (\ref{P3KKTeqStationaritypn}) and~(\ref{P3KKTeqStationarityBn}), we have
\begin{talign}
(\nabla_{p_n} L_{\mathbb{P}_3}(\boldsymbol{p}, \boldsymbol{B}, \boldsymbol{\tau}, \lambda \mid \boldsymbol{\beta},\boldsymbol{\nu}))|_{\begin{subarray}{l}\bm{p}=\bm{p}^{\#}(\bm{\beta},\bm{\nu}), \\ \bm{b}=\bm{b}^{\#}(\bm{\beta},\bm{\nu})\end{subarray}} & = 0, \text{ and}  \nonumber \\ \nabla_{b_n} L_{\mathbb{P}_3}(\boldsymbol{p}, \boldsymbol{B}, \boldsymbol{\tau}, \lambda \mid \boldsymbol{\beta},\boldsymbol{\nu}))|_{\begin{subarray}{l}\bm{p}=\bm{p}^{\#}(\bm{\beta},\bm{\nu}), \\ \bm{b}=\bm{b}^{\#}(\bm{\beta},\bm{\nu})\end{subarray}} & =0  ,\nonumber
\end{talign}
which are used in (\ref{eqlp3fderiv2}) to obtain
\begin{talign}
& \nabla_{\bm{\alpha}} \left(L_{\mathbb{P}_3}(\bm{p}^{\#}(\bm{\beta},\bm{\nu}), \bm{B}^{\#}(\bm{\beta},\bm{\nu}), \boldsymbol{\tau}, \lambda \mid \boldsymbol{\beta},\boldsymbol{\nu})\right)\nonumber \\ & =   \left[(\nabla_{\bm{\alpha}} L_{\mathbb{P}_3}(\boldsymbol{p}, \boldsymbol{B}, \boldsymbol{\tau}, \lambda \mid \boldsymbol{\beta},\boldsymbol{\nu}))|_{\begin{subarray}{l}\bm{p}=\bm{p}^{\#}(\bm{\beta},\bm{\nu}), \\ \bm{b}=\bm{b}^{\#}(\bm{\beta},\bm{\nu})\end{subarray}}\right].\label{eqlp3fderiv3}
\end{talign}
From (\ref{eqLP3}) (\ref{eqlp3fderi}) and (\ref{eqlp3fderiv3}), it holds that
\begin{talign}
& \nabla_{\bm{\alpha}} \left(\sum_{n\in \mathcal{N}} \mathcal{F}_n(p_n^{\#}(\bm{\beta},\bm{\nu}), B_n^{\#}(\bm{\beta},\bm{\nu})|\beta_n, \nu_n)\right)\nonumber \\ & =   -\left[(\nabla_{\bm{\alpha}} L_{\mathbb{P}_3}(\boldsymbol{p}, \boldsymbol{B}, \boldsymbol{\tau}, \lambda \mid \boldsymbol{\beta},\boldsymbol{\nu}))|_{\begin{subarray}{l}\bm{p}=\bm{p}^{\#}(\bm{\beta},\bm{\nu}), \\ \bm{b}=\bm{b}^{\#}(\bm{\beta},\bm{\nu})\end{subarray}}\right]\nonumber \\ & =  \sum_{n\in \mathcal{N}}\left[\left( \nabla_{\bm{\alpha}} \mathcal{F}_n(p_n, B_n|\beta_n, \nu_n)\right)|_{\begin{subarray}{l}\bm{p}=\bm{p}^{\#}(\bm{\beta},\bm{\nu}), \\ \bm{b}=\bm{b}^{\#}(\bm{\beta},\bm{\nu})\end{subarray}}\right]. \label{eqlp3fderiv4}
\end{talign}
Using the above and (\ref{mathcalF1}), we further acquire 
\begin{talign}
& \nabla_{\beta_n} \left(\sum_{n\in \mathcal{N}} \mathcal{F}_n(p_n^{\#}(\bm{\beta},\bm{\nu}), B_n^{\#}(\bm{\beta},\bm{\nu})|\beta_n, \nu_n)\right)\nonumber \\ & =  \sum_{n\in \mathcal{N}}\left[\left( \nabla_{\beta_n} \mathcal{F}_n(p_n, B_n|\beta_n, \nu_n)\right)|_{\begin{subarray}{l}\bm{p}=\bm{p}^{\#}(\bm{\beta},\bm{\nu}), \\ \bm{b}=\bm{b}^{\#}(\bm{\beta},\bm{\nu})\end{subarray}}\right]\nonumber \\ & = - \sum_{n\in \mathcal{N}}\left[ \nu_n\cdot  (p_n^{\#}(\bm{\beta},\bm{\nu})+p_n^{\textnormal{cir}})\right]\label{eqlp3fderiv5},
\end{talign}
and
\begin{talign}
& \nabla_{\nu_n} \left(\sum_{n\in \mathcal{N}} \mathcal{F}_n(p_n^{\#}(\bm{\beta},\bm{\nu}), B_n^{\#}(\bm{\beta},\bm{\nu})|\beta_n, \nu_n)\right)\nonumber \\ & =  \sum_{n\in \mathcal{N}}\left[\left( \nabla_{\nu_n} \mathcal{F}_n(p_n, B_n|\beta_n, \nu_n)\right)|_{\begin{subarray}{l}\bm{p}=\bm{p}^{\#}(\bm{\beta},\bm{\nu}), \\ \bm{b}=\bm{b}^{\#}(\bm{\beta},\bm{\nu})\end{subarray}}\right]\nonumber \\ & =  \sum_{n\in \mathcal{N}}\left[ F_n(p_n^{\#}(\bm{\beta},\bm{\nu}), B_n^{\#}(\bm{\beta},\bm{\nu})) -\beta_n \cdot (p_n^{\#}(\bm{\beta},\bm{\nu})+p_n^{\textnormal{cir}})\right].\label{eqlp3fderiv6}
\end{talign}
Since (\ref{eqlp3fderiv5}) and (\ref{eqlp3fderiv6}) hold for any $ \bm{\beta} $ and $ \bm{\nu}$, we obtain
\begin{talign}
& \nabla_{\beta_n}   \mathcal{F}_n(p_n^{\#}(\bm{\beta},\bm{\nu}), B_n^{\#}(\bm{\beta},\bm{\nu})|\beta_n, \nu_n)  \nonumber \\ & = -  \left[ \nu_n\cdot  (p_n^{\#}(\bm{\beta},\bm{\nu})+p_n^{\textnormal{cir}})\right]\label{eqlp3fderiv5v2},
\end{talign}
and
\begin{talign}
& \nabla_{\nu_n}   \mathcal{F}_n(p_n^{\#}(\bm{\beta},\bm{\nu}), B_n^{\#}(\bm{\beta},\bm{\nu})|\beta_n, \nu_n) \nonumber \\ & =  \left[ F_n(p_n^{\#}(\bm{\beta},\bm{\nu}), B_n^{\#}(\bm{\beta},\bm{\nu})) -\beta_n \cdot (p_n^{\#}(\bm{\beta},\bm{\nu})+p_n^{\textnormal{cir}})\right].\label{eqlp3fderiv6v2}
\end{talign}
From the definition of $\mathcal{F}_n(p_n, B_n \,| \, \beta_n, \nu_n)$ in (\ref{mathcalF1}), we also have
\begin{talign}
& \nabla_{\beta_n}   \mathcal{F}_n(p_n^{\#}(\bm{\beta},\bm{\nu}), B_n^{\#}(\bm{\beta},\bm{\nu})|\beta_n, \nu_n)  \nonumber \\ & = \nu_n\cdot \big[\nabla_{p_n} F_n(p_n, B_n)|_{p_n=p_n^{\#}(\bm{\beta},\bm{\nu})} \cdot \nabla_{\beta_n} p_n^{\#}(\bm{\beta},\bm{\nu})  \nonumber \\ & \quad + \nabla_{B_n} F_n(p_n, B_n)|_{B_n=B_n^{\#}(\bm{\beta},\bm{\nu})} \cdot \nabla_{\beta_n} B_n^{\#}(\bm{\beta},\bm{\nu})  - (p_n^{\#}(\bm{\beta},\bm{\nu})+p_n^{\textnormal{cir}})\big]\label{eqlp3fderiv5v22},
\end{talign}
and
\begin{talign}
& \nabla_{\nu_n}   \mathcal{F}_n(p_n^{\#}(\bm{\beta},\bm{\nu}), B_n^{\#}(\bm{\beta},\bm{\nu})|\beta_n, \nu_n) \nonumber \\ & =  \left[ F_n(p_n^{\#}(\bm{\beta},\bm{\nu}), B_n^{\#}(\bm{\beta},\bm{\nu})) -\beta_n \cdot (p_n^{\#}(\bm{\beta},\bm{\nu})+p_n^{\textnormal{cir}})\right] \nonumber \\ & \quad +\nu_n\cdot \big[\nabla_{p_n} F_n(p_n, B_n)|_{p_n=p_n^{\#}(\bm{\beta},\bm{\nu})} \cdot \nabla_{\nu_n} p_n^{\#}(\bm{\beta},\bm{\nu})  \nonumber \\ & \quad + \nabla_{B_n} F_n(p_n, B_n)|_{B_n=B_n^{\#}(\bm{\beta},\bm{\nu})} \cdot \nabla_{\nu_n} B_n^{\#}(\bm{\beta},\bm{\nu}) \nonumber \\ & \quad - \beta_n \cdot \nabla_{\nu_n} p_n^{\#}(\bm{\beta},\bm{\nu}) \big].\label{eqlp3fderiv6v22}
\end{talign}
Comparing (\ref{eqlp3fderiv5v2}) and (\ref{eqlp3fderiv5v22}), and comparing (\ref{eqlp3fderiv6v2}) and (\ref{eqlp3fderiv6v22}), since we always enforce $\nu_n > 0$, we have proved
\begin{talign}
& \big[ \nabla_{p_n} F_n(p_n, B_n)|_{p_n=p_n^{\#}(\bm{\beta},\bm{\nu})} \cdot \nabla_{\beta_n} p_n^{\#}(\bm{\beta},\bm{\nu})  \nonumber \\ &   + \nabla_{B_n} F_n(p_n, B_n)|_{B_n=B_n^{\#}(\bm{\beta},\bm{\nu})} \cdot \nabla_{\beta_n} B_n^{\#}(\bm{\beta},\bm{\nu}) \big]   =0 \label{eqlp3fderiv5v223},
\end{talign}
and
\begin{talign}
& \big[\nabla_{p_n} F_n(p_n, B_n)|_{p_n=p_n^{\#}(\bm{\beta},\bm{\nu})} \cdot \nabla_{\nu_n} p_n^{\#}(\bm{\beta},\bm{\nu})  \nonumber \\ &   + \nabla_{B_n} F_n(p_n, B_n)|_{B_n=B_n^{\#}(\bm{\beta},\bm{\nu})} \cdot \nabla_{\nu_n} B_n^{\#}(\bm{\beta},\bm{\nu}) - \beta_n \cdot \nabla_{\nu_n} p_n^{\#}(\bm{\beta},\bm{\nu})\big] = 0.\label{eqlp3fderiv6v223}
\end{talign}
Based on the above, and the definitions of $\phi_{1,n}(\bm{\beta},\bm{\nu})$ and $\phi_{2,n}(\bm{\beta},\bm{\nu})$ in~(\ref{eqn:phi_1v2}) and~(\ref{eqn:phi_2v2}), the desired result (\ref{bothequalv2}) is proved.

\section[]{Solving Problem $\mathbb{P}_3(\bm{\beta},\bm{\nu})$}\label{Appendix:theorem6.3}

We will prove Theorem~\ref{thmsolveP3} and use it to solve Problem $\mathbb{P}_3(\bm{\beta},\bm{\nu})$.

\subsection[]{Proof of Theorem~\ref{thmsolveP3} which characterizes the solution to Problem $\mathbb{P}_3(\bm{\beta},\bm{\nu})$}\label{Appendix:theorem6.3thm}

Some texts below are repeated from Appendix~\ref{appbothequal}. 
Problem $\mathbb{P}_3(\bm{\beta},\bm{\nu})$ belongs to convex optimization and Slater's condition holds, as shown in the proof of Lemma~\ref{lemma:sum-of-ratios_lemma}. Then the Karush--Kuhn--Tucker (KKT) conditions are necessary and sufficient to obtain the globally optimal solution, as stated in Lemma~\ref{lemmacvxuseKKT}.

Let $[\boldsymbol{p}^{\#}, \boldsymbol{B}^{\#}, \boldsymbol{\tau}^{\#}, \lambda^{\#}]$ satisfy the KKT conditions of Problem $\mathbb{P}_3(\bm{\beta},\bm{\nu})$. Then, after defining
\begin{align} 
\textstyle{\vartheta_n^{\#} : = \frac{g_np_n^{\#}}{\sigma_n^2B_n^{\#}}},\label{definevartheta}
\end{align}
we obtain the following KKT conditions:
 \begin{talign}
     \frac{\partial L_{\mathbb{P}_3}}{\partial p_n} &= -(\nu_n c_nf_n^{\prime}(r_{n,\textnormal{s}}(p_n^{\#},B_n^{\#}))+ \tau_n^{\#})\frac{g_n}{\sigma_n^2(1+\vartheta_n^{\#})\ln{2}} \nonumber \\ & \quad +\nu_n\beta_n    = 0, ~\forall n\in \mathcal{N},\label{kkt:partial_p}\\
     \frac{\partial L_{\mathbb{P}_3}}{\partial B_n} &= 
     -\nabla_{B_n} \mathcal{F}_n(p_n, B_n \,| \, \beta_n, \nu_n) |_{B_n = B_n^{\#},p_n = p_n^{\#}} \nonumber \\ 
     &\quad~- \tau_n^{\#} \nabla_{B_n} r_n(p_n,B_n)|_{B_n = B_n^{\#},p_n = p_n^{\#}} \!+\! \lambda^{\#} \label{kkt:partial_bint}\\ &=
     -\Big(\nu_n c_n f_n^{\prime}(r_{n,\textnormal{s}}(p_n^{\#},B_n^{\#})) + \tau_n^{\#} \Big)\Big(\log_2(1 \!+\! \vartheta_n^{\#})\notag\\
     & \quad~- \frac{\vartheta_n^{\#}}{(1+\vartheta_n^{\#})\ln{2}}\Big) \!+\! \lambda^{\#} \!=\! 0, ~ \forall n\in \mathcal{N},\label{kkt:partial_b}\\
     &\lambda^{\#}\cdot(\sum_{n\in \mathcal{N}} B_n^{\#} -B_{\text{total}}) = 0, \label{kkt:lambda} \\
     & \tau_n^{\#}\cdot(r_n^{\min} -r_n(p_n^{\#},B_n^{\#}) )=0,~ \forall n\in \mathcal{N},\label{kkt:tau}
 \end{talign}
where (\ref{kkt:partial_p}) and (\ref{kkt:partial_b}) refer to the stationarity conditions, while (\ref{kkt:lambda}) and (\ref{kkt:tau}) are called complementary slackness. We show the intermediate step (\ref{kkt:partial_bint}) since it will be useful later. For the conditions of primal feasibility (i.e., (\ref{constra:bandwidth}) and~(\ref{constra:rate}) for $[\boldsymbol{p}^{\#}, \boldsymbol{B}^{\#}]$) and dual feasibility (i.e., $\tau_n^{\#} \geq 0$ for all $n \in \mathcal{N}$ and $\lambda^{\#} \geq 0$), we will write them out at the places where we need them.

Next, we aim to simplify (\ref{kkt:partial_p})-(\ref{kkt:tau}) step-by-step to obtain $[\boldsymbol{p}^{\#}, \boldsymbol{B}^{\#}, \boldsymbol{\tau}^{\#}, \lambda^{\#}]$. To begin with, Condition~\ref{fconcave} on Page~\pageref{fconcave} ensures \\ $f^{\prime}(r_{n,\textnormal{s}}(p_n^{\#},B_n^{\#}))>0$. Using this along with $\tau_n^{\#} \geq 0$ and\footnote{Using $  r_n(p_n^{\#},B_n^{\#}) \geq r_n^{\min}$ from (\ref{constra:rate}) and $r_n^{\min} >0$ from Condition~\ref{conditionrn} on Page~\pageref{conditionrn}, we have $  r_n(p_n^{\#},B_n^{\#}) > 0$ which implies $p_n^{\#}> 0$ and $B_n^{\#}> 0$, inducing $\vartheta_n^{\#} >0 $. For any $x>0$, we can prove $\log_2(1 + x) - \frac{x}{(1+x)\ln{2}} > 0$, so that $\log_2(1 + \vartheta_n^{\#}) - \frac{\vartheta_n^{\#}}{(1+\vartheta_n^{\#})\ln{2}} > 0$. \label{footnotepstarBstar}} $\log_2(1 + \vartheta_n^{\#}) - \frac{\vartheta_n^{\#}}{(1+\vartheta_n^{\#})\ln{2}} > 0$ in (\ref{kkt:partial_b}), we know $\lambda^{\#} >0$ so that (\ref{kkt:lambda}) becomes
\begin{talign} 
\sum_{n\in \mathcal{N}} B_n^{\#} = B_{\text{total}}. \label{kkt:lambda2}
\end{talign}

We note that both (\ref{kkt:partial_p}) and (\ref{kkt:partial_b}) have the term \\ $\nu_n c_nf^{\prime}(r_{n,\textnormal{s}}(p_n^{\#},B_n^{\#}))+ \tau_n^{\#}$, which is strictly positive due to $\tau_n^{\#} \geq 0$ and $f^{\prime}(r_{n,\textnormal{s}}(p_n^{\#},B_n^{\#}))>0$ explained above. Thus, from (\ref{kkt:partial_p}) and (\ref{kkt:partial_b}), we get
\begin{talign} 
\frac{
    \Big(\log_2\big(1+\vartheta_n^{\#}\big)
     -\frac{\vartheta_n^{\#}}{(1+\vartheta_n^{\#})\ln2}\Big)\nu_n  }{
    \frac{g_n}{\sigma_n^2(1+\vartheta_n^{\#})\ln2}} = \frac{\lambda^{\#}}{ \beta_n}. \label{solvepsin}
\end{talign}

From the above equation (\ref{solvepsin}), we   solve $\vartheta_n^{\#}$ given $\lambda^{\#}$. Denoting the solution as $\psi_n(\lambda)$ to highlight its dependence on $\lambda$, we have:
\begin{talign}
    \vartheta_n^{\#} \hspace{-1pt}=\hspace{-1pt} \psi_n(\lambda^{\#}), \textnormal{ for } \textstyle{\psi_n(\lambda) \hspace{-1pt}:=\hspace{-1pt} \exp\hspace{-1pt}\big\{\hspace{-1pt}1\hspace{-1pt}+\hspace{-1pt}W\big(\frac{1}{e}(\frac{g_n\lambda}{\nu_n\beta_n\sigma_n^2}\hspace{-1pt}-\hspace{-1pt}1)\big)\hspace{-1pt}\big\}\hspace{-1pt}-\hspace{-1pt}1}. \label{solvepsin2}
\end{talign}

Once we have $\lambda^{\#}$,
(\ref{definevartheta}) and~(\ref{solvepsin2}) mean that $\vartheta_n^{\#} $ denoting $\frac{g_np_n^{\#}}{\sigma_n^2B_n^{\#}}$ is decided. To derive $p_n^{\#}$ and $B_n^{\#}$ given $\lambda^{\#}$, we need another condition of $p_n^{\#}$ and $B_n^{\#}$. To this end, we notice (\ref{kkt:partial_b}), but (\ref{kkt:partial_b}) invovles  $\tau_n^{\#} \geq 0$. If $\tau_n^{\#} = 0$, then (\ref{kkt:partial_b}) together with (\ref{solvepsin2}) will decide $p_n^{\#}$ and $B_n^{\#}$ given $\lambda^{\#}$. Therefore, we will discuss \textbf{Case 1:} $\tau_n^{\#} = 0$ and \textbf{Case 2:} $\tau_n^{\#} > 0$ respectively for each $n\in \mathcal{N}$. From the above explanation, we first try to express $p_n^{\#}$ and $B_n^{\#}$ as expressions of $\lambda^{\#}$, and then substitute these expressions into our conditions to obtain $\lambda^{\#}$.

Before elaborating on the two cases, we note (\ref{kkt:partial_bint}) and define a function $\gamma_n( \lambda)$ which will facilitate discussing the two cases. Specifically, given $\lambda$, then under the constraint of 
\begin{align}
\textstyle{\frac{g_np_n}{\sigma_n^2B_n}=\psi_n(\lambda),}~\textnormal{ for $\psi_n(\lambda)$ defined in~(\ref{solvepsin2})},\label{psiconstraint}
\end{align}
we define $\gamma_n( \lambda)$ as the result of $r_n(p_n,B_n)\geq r_{n, \textnormal{e}}$ to ensure (we will discuss soon when such $r_n(p_n,B_n)$ does not exist)
\begin{talign}
\nabla_{B_n} \mathcal{F}_n(p_n, B_n \,| \, \beta_n, \nu_n)  = \lambda.\label{nablaconstraint}
\end{talign} 

From $r_n(p_n,B_n) = B_n\log_2(1+\frac{g_n p_n}{{\sigma_n}^2B_n})$, $B_n$ ensuring (\ref{psiconstraint}) and~(\ref{nablaconstraint}) is given by
\begin{talign}
\textstyle{\frac{\gamma_n( \lambda)}{\log_2(1+\psi_n(\lambda))}}   .\label{nablaconstraintBn}
\end{talign} 

We aim to obtain the expression of $\gamma_n( \lambda)$ for $r_n(p_n,B_n)$ from~(\ref{psiconstraint}) and~(\ref{nablaconstraint}). From (\ref{mathcalF1}) and~(\ref{nablaconstraint}),
\begin{talign}
 &f_n^{\prime}(r_n(p_n,B_n) - r_{n, \textnormal{e}})   \nonumber \\ &= \textstyle{\frac{\lambda}{\nu_n c_n \cdot (\log_2\big(1+\psi_n(\lambda)\big)
     -\frac{\psi_n(\lambda)}{(1+\psi_n(\lambda))\ln2})}}  = \textstyle{\frac{\beta_n\sigma_n^2(1+\psi_n(\lambda))\ln{2}}{c_ng_n}} ,\label{correspondingBnstarx}
\end{talign}
where the last step uses $\frac{
    (\log_2(1+\psi_n(\lambda))
     -\frac{\psi_n(\lambda)}{(1+\psi_n(\lambda))\ln2})\nu_n  }{
    \frac{g_n}{\sigma_n^2(1+\psi_n(\lambda))\ln2}} = \frac{\lambda}{ \beta_n}$ from~(\ref{solvepsin}) and~(\ref{solvepsin2}).

From~(\ref{correspondingBnstarx}), we know $r_n(p_n,B_n)\geq r_{n, \textnormal{e}}$ ensuring~(\ref{nablaconstraint}) may not exist for all $\lambda$, since we do not know the range of $f_n^{\prime}$. Whenever such $r_n(p_n,B_n)\geq r_{n, \textnormal{e}}$ does not exist, we just define $\gamma_n( \lambda)$ as $r_{n, \textnormal{e}}$. The above leads to the desired expression of $\gamma_n( \lambda)$ in  Theorem~\ref{thmsolveP3} on Page~\pageref{thmsolveP3}.


We now discuss the two cases for each $n \in \mathcal{N}$:
\begin{itemize}
\item \textbf{Case 1:} $\tau_n^{\#} = 0$. In this case, (\ref{definevartheta}) (\ref{kkt:partial_bint}) and (\ref{solvepsin2}) mean that setting $p_n$, $B_n$, and $\lambda$ as $p_n^{\#}$, $B_n^{\#}$, and $\lambda^{\#}$ respectively ensures~(\ref{psiconstraint}) and~(\ref{nablaconstraint}), where
we note the primal feasibility condition~(\ref{constra:rate}) along with Condition~\ref{conditionrn} on Page~\pageref{conditionrn} means $r_n(p_n^{\#},B_n^{\#}) \geq r_{n}^{\min} \geq r_{n, \textnormal{e}}$. Noting the above and~(\ref{nablaconstraintBn}),
 we obtain $B_{n}^{\#} =  \frac{\gamma_n( \lambda^{\#})}{\log_2(1+\psi_n(\lambda^{\#}))} \geq \frac{r_{n}^{\min}}  {\log_2(1+\psi_n(\lambda^{\#}))}$. 
\item \textbf{Case 2:} $\tau_n^{\#} > 0$. In this case, (\ref{kkt:tau}) means $r_n(p_n^{\#},B_n^{\#})  = r_{n}^{\min}$, which along with (\ref{definevartheta}) and~(\ref{solvepsin2}) induces $B_{n}^{\#} = \frac{r_{n}^{\min}}  {\log_2(1+\psi_n(\lambda^{\#}))}$. Also,  (\ref{kkt:partial_bint}) means 
\begin{talign}
& \nabla_{B_n} \mathcal{F}_n(p_n, B_n \,| \, \beta_n, \nu_n) |_{B_n = B_n^{\#},p_n = p_n^{\#}} \nonumber \\ &  = \lambda^{\#} - \tau_n^{\#} \nabla_{B_n} r_n(p_n,B_n)|_{B_n = B_n^{\#}} <  \lambda^{\#} ,\label{nablaconstraintlam}
\end{talign} 
where the last step uses $\tau_n^{\#} > 0$ and $\nabla_{B_n} r_n(p_n,B_n)|_{B_n = B_n^{\#}}>0$ (note $B_n^{\#}>0$ as explained in Footnote~\ref{footnotepstarBstar}). The above means setting $p_n$, $B_n$, and $\lambda$ as $p_n^{\#}$, $\frac{r_{n}^{\min}}  {\log_2(1+\psi_n(\lambda^{\#}))}$, and $\lambda^{\#}$ respectively ensures~(\ref{psiconstraint}) and~(\ref{nablaconstraintlam}). Moreover, when $\gamma_n( \lambda^{\#})>0$ exists, (\ref{definevartheta}) (\ref{solvepsin2})  and~(\ref{nablaconstraintBn}) means setting $p_n$, $B_n$, and $\lambda$ as $p_n^{\#}$, $\frac{\gamma_n( \lambda^{\#})}  {\log_2(1+\psi_n(\lambda^{\#}))}$, and $\lambda^{\#}$ respectively ensures~(\ref{psiconstraint}) and~(\ref{nablaconstraint}). Comparing(\ref{nablaconstraint}) and~(\ref{nablaconstraintlam}), and noting that \\$\nabla_{B_n} \mathcal{F}_n(p_n, B_n \,| \, \beta_n, \nu_n) = \nu_n c_n  f_n^{\prime}(r_{n,\textnormal{s}}(p_n,B_n)) \nabla_{B_n}r_{n,\textnormal{s}}(p_n,B_n) $ is strictly decreasing function with respect to\footnote{This holds since $\nabla_{B_n}r_{n,\textnormal{s}}(p_n,B_n)$ is strictly decreasing  with respect to $B_n$ and positive, and $f_n^{\prime}(r_{n,\textnormal{s}}(p_n,B_n))$ is \mbox{non-increasing}  with respect to $B_n$ and positive given Condition~\ref{fconcave} on Page~\pageref{fconcave}.} $B_n$, we obtain $B_n^{\#} =\frac{r_{n}^{\min}}  {\log_2 (1+\psi_n(\lambda^{\#}))} > \frac{\gamma_n( \lambda^{\#}) }  {\log_2(1+\psi_n(\lambda^{\#}))}   $. When we cannot find $\gamma_n( \lambda^{\#}) \geq r_{n, \textnormal{e}}$ for (\ref{nablaconstraint}), as already explained, we just set $\gamma_n( \lambda^{\#})$ as $r_{n, \textnormal{e}}$ and still have $B_n^{\#} =\frac{r_{n}^{\min}}  {\log_2 (1+\psi_n(\lambda^{\#}))} \geq \frac{\gamma_n( \lambda^{\#}) }  {\log_2(1+\psi_n(\lambda^{\#}))}   $ due to $r_{n}^{\min} \geq r_{n, \textnormal{e}}$. 
\end{itemize}

Summarizing the two cases, we conclude for any $n \in \mathcal{N}$ that
\begin{talign}
B_n^{\#} & = \textstyle{\max\{ \frac{\gamma_n( \lambda^{\#})}{\log_2\big(1+\psi_n(\lambda^{\#}))},~ \frac{r_{n}^{\min}}  {\log_2(1+\psi_n(\lambda^{\#}))}\}} \nonumber \\ & =  \textstyle{\frac{\max\{ \gamma_n( \lambda^{\#}),~ r_{n}^{\min}\} }{\log_2(1+\psi_n(\lambda^{\#}))}}. \label{correspondingBnstar}
\end{talign}




Then~(\ref{correspondingBnstarx}) and $r_{n,\textnormal{s}}(p_n,B_n) =r_n(p_n,B_n) -r_{n, \textnormal{e}}$  Now we know how to compute $B_n^{\#}$ in~(\ref{correspondingBnstar}) given $\lambda^{\#}$. Then $\lambda^{\#}$ is decided such that $B_n^{\#}|_{n\in \mathcal{N}}$ from~(\ref{correspondingBnstar}) together satisfy~(\ref{kkt:lambda2}). Finally, after $\lambda^{\#}$ and $B_n^{\#}$ are obtained, $p_n^{\#}$ is computed as $\frac{\sigma_n^2B_n^{\#}\cdot \psi_n(\lambda^{\#})}{g_n}$ based on (\ref{definevartheta}) and~(\ref{solvepsin2}). To summarize, we have proved Eq.~(\ref{eqoptimal}) of Theorem~\ref{thmsolveP3} on Page~\pageref{thmsolveP3}; i.e., {Theorem~\ref{thmsolveP3} is proved.} 
\qed

For strictly concave utility, Lemma~\ref{lemb.1} below shows that $\mathbb{P}_3(\bm{\beta},\bm{\nu})$ has a unique globally optimal solution given by Theorem~\ref{thmsolveP3}.

\begin{lemma} \label{lemb.1}
In Theorem~\ref{thmsolveP3}, with an additional condition that the function $f_n(x)$ for any $n \in \mathcal{N}$ is strictly concave, then Theorem~\ref{thmsolveP3} gives the unique globally optimal solution of $\mathbb{P}_3(\bm{\beta},\bm{\nu})$. 
\end{lemma}
\begin{proof}  
We will show the following three results, where $\Psi(\lambda)$ denotes $\sum_{n \in \mathcal{N}} \mathcal{B}_n(\lambda)$.
  \begin{itemize}
      \item[\textbf{\ding{182}}] $\Psi(\lambda)$ is strictly decreasing in $\lambda \geq 0$,
      \item[\textbf{\ding{183}}] $\lim_{\lambda \to 0^+}\Psi(\lambda) = \infty$, and
      \item[\textbf{\ding{184}}] $\lim_{\lambda \to \infty}\Psi(\lambda) = 0$.
  \end{itemize} 

\textbf{Proving Result \ding{182}:}
In Eq.~(\ref{eqoptimal}), we also define function $\mathcal{B}_n(\lambda)$. It is used to better explain the proof here for Result (ii). Note we always enforce $\lambda \geq 0$ below.  Here we consider the function $f_n(\cdot)$ for any $n \in \mathcal{N}$ is strictly concave (i.e., $f_n'(\cdot)$ is strictly decreasing) for $x \geq 0$. Combining this with the fact that $\psi_n(\lambda) $ in (\ref{solvepsin2}) is increasing in $\lambda$, we know that $\gamma_n(\lambda)$ in~(\ref{eqoptimal}) is  \mbox{non-increasing}  in $\lambda$. Since $\psi_n(\lambda) $ is  increasing and $\gamma_n(\lambda)$ is  decreasing, $\mathcal{B}_n(\lambda)$ is strictly decreasing in $\lambda$, so that Result \ding{182} is proved. 

\textbf{Proving Results \ding{183} and \ding{184}:}
Since $f_n'(x)$ here is strictly decreasing, $\gamma_n(\lambda)$ defined in~(\ref{eqoptimal}) is at most $ r_{n, \textnormal{e}}+ \zeta$ for $\zeta:=(f_n^{\prime})^{-1}\big(\frac{\beta_n\sigma_n^2\ln{2}}{c_ng_n}\big)$ when
$\zeta \geq 0$ exists, and equals $ r_{n, \textnormal{e}} $ otherwise. Anyways, $\gamma_n(\lambda)$ is upper bounded by a constant. From~(\ref{solvepsin2}), we know $\lim_{\lambda \to 0^+}\psi_n(\lambda) = 0$ and $\lim_{\lambda \to \infty}\psi_n(\lambda) = \infty$. Then from~(\ref{eqoptimal}), we have $\lim_{\lambda \to 0^+}\mathcal{B}_n(\lambda) = 0$ and $\lim_{\lambda \to \infty}\mathcal{B}_n(\lambda) = \infty$, so that Results \ding{183} and \ding{184} are proved.

As noted in~(\ref{solvepsin2}), $\lambda^{\#}$ is the solution of $\lambda$ to $\sum\limits_{n \in \mathcal{N}} \mathcal{B}_n(\lambda) = B_{\text{total}}$. Clearly, $\lambda^{\#}$ is unique given Results \ding{182} \ding{183} and \ding{184} above. 
Thus, the desired result is proved.
\end{proof}




\subsection[]{Algorithm to solve Problem $\mathbb{P}_3(\bm{\beta},\bm{\nu})$ based on Theorem~\ref{thmsolveP3}} \label{applambdastar}



   We still let $\Psi(\lambda)$ denote $\sum_{n \in \mathcal{N}} \mathcal{B}_n(\lambda)$. Similar to Lemma~\ref{lemb.1} for strictly concave utility, we can prove   for  concave utility, 
$\Psi(\lambda)$ is \mbox{non-increasing} as $\lambda  $ increases. This motivates us to
use the bisection method to find $\lambda^{\#}$ from~(\ref{lambdastar}).

For the  bisection method, we use $0$ as the initial lower bound. We can find the initial upper bound as follows. Starting with a random positive number $\Lambda$. If $\Psi(\Lambda)$ denoting $\sum_{n \in \mathcal{N}} \mathcal{B}_n(\Lambda)$ is  less than $ B_{\textnormal{total}}$, we use $\Lambda$ as the initial upper bound. If $\Psi(\Lambda)$ equals $ B_{\textnormal{total}}$, then $\Lambda$ is just our desired $\lambda^{\#}$. If $\Psi(\Lambda)$  is  greater than $ B_{\textnormal{total}}$, we check $\Psi(2\Lambda)$ denoting $\sum_{n \in \mathcal{N}} \mathcal{B}_n(2\Lambda)$. Similarly, if $\Psi(2\Lambda)$ is  less than $ B_{\textnormal{total}}$, we use $2\Lambda$ as the initial upper bound. If $\Psi(2\Lambda)$  equals $ B_{\textnormal{total}}$, then $2\Lambda$ is just our desired $\lambda^{\#}$. If $\Psi(2\Lambda)$ is  greater than $ B_{\textnormal{total}}$, we check $\Psi(2^2\Lambda)$ denoting $\sum_{n \in \mathcal{N}} \mathcal{B}_n(2^2\Lambda)$. The process continues. Basically, we find $i$ such that $\Psi(2^{i-1}\Lambda)$ is  greater than $ B_{\textnormal{total}}$, and $\Psi(2^i\Lambda)$ is  less than $ B_{\textnormal{total}}$, then we use $2^i\Lambda$ as the initial upper bound. If there exists $i$ which makes $\Psi(2^i\Lambda)$ equal $ B_{\textnormal{total}}$, then $2^i\Lambda$ is just our desired $\lambda^{\#}$.

With the initial lower bound and the  initial upper bound explained above, the remaining process to  find $\lambda^{\#}$ follows from the standard bisection method. In each iteration, the bisection method divides the interval $[a,b]$ in two parts  by computing the midpoint $c = (a+b) / 2$ of the interval and the value of   $\Psi(c)$. If $\Psi(c)$ equals $ B_{\textnormal{total}}$, then the process has succeeded and $c$ is just our desired $\lambda^{\#}$. Otherwise, 
if $\Psi(c)$ is greater than $ B_{\textnormal{total}}$, we update $a$ to $c$ so that the next iteration starts with the interval $[c,b]$; if $\Psi(c)$ is less than $ B_{\textnormal{total}}$, we update $b$ to $c$ so that the next iteration starts with the interval $[a,c]$. The  bisection method   converges when $\Psi(c)$ is close to (but should be no greater than) $ B_{\textnormal{total}}$.
 

Note that for each $\lambda$, computing $\Psi(\lambda)$ denoting $\sum_{n \in \mathcal{N}} \mathcal{B}_n(\lambda)$ costs $\mathcal{O}(N)$ time. The number of iterations to find the initial upper bound depends on the initialization, while the number of iterations for the bisection method depends on the error tolerance. With the initial lower bound $0$, the  initial upper bound $H$, and the error tolerance $\epsilon$, the number of iterations for the bisection method is $\mathcal{O}(\log_2{\frac{H}{\epsilon}})$.

\interdisplaylinepenalty=10000

\section{Baseline Algorithms for Comparison}\label{sec:baseline model}

As shown in Section~\ref{secComparison}, we compare our Algorithm~\ref{algo:MN} with the following baselines: ``Optimizing $\bm{B}$ only'', ``Optimizing $\bm{p}$ only'', and ``Alternating optimization''. We detail them below. 


\subsection[]{Optimizing $\bm{p}$ only}\label{algo:p_n_only}

Given $B_n$, we define $p_n^{\min}$ as the value of $p_n$ which causes $r_n(p_n,B_n)$ to be $r_n^{\min}$. Formally, 
\begin{talign}
 &\text{$p_n^{\min}:=\frac{(2^{\frac{r_n^{\min}}{B_n}}-1){\sigma_n}^2B_n}{g_n}$ so that $B_n\log_2(1+\frac{g_n p_n^{\min}}{{\sigma_n}^2B_n}) = r_n^{\min}$.} \label{eqdefpnmin}  
\end{talign}

Then  $r_n \geq  r_n^{\min}  $ in (\ref{constra:rate}) of Problem $\mathbb{P}_{1}$ means $p_n \geq p_n^{\min}$. Then ``Optimizing $\bm{p}$ only'' just means for each $n \in \mathcal{N}$, maximizing $\varphi_n(p_n,B_n)$ subject to $p_n \geq p_n^{\min}$.

\begin{lemma} \label{lemma-quasiconvex}
For each $n \in \mathcal{N}$, the following results 
 hold. 
\begin{itemize}
\item[(i)] $\varphi_n(p_n,B_n)$ is pseudoconvex and semistrictly~quasiconvex for $p_n \in  [p_n^{\min}, \infty)$ and $B_n\in(0, \infty)$.
\item[(ii)] To optimize the  weighted sum-UEE (which means maximizing $\varphi_n(p_n,B_n)$ given $B_n$ for $p_n \in  [p_n^{\min}, \infty)$), we just need to compute a stationary point of $\varphi_n(p_n,B_n)$ with respect to $p_n$ given $B_n$ and have a comparison with $p_n^{\min}$. The maximum of them will be a point at which $\varphi_n(p_n,B_n)$ achieves the maximum.
\end{itemize}
\end{lemma}

Given Section~\ref{secTypicalUtility}, below we consider $b_n=1$, $c_n=0$, and $d_n=0$     so that the three types of utility functions are as follows:
\begin{itemize}
\item Type 1 utility function: 
$f_n(x)=\kappa_n \ln (1+ a_n x)$ with $\kappa_n >0$ and $0<a_n<1$,
\item Type 2 utility function: 
\mbox{$f_n(x)= \kappa_n \cdot (1- e^{-a_n x})$} with $a_n, \kappa_n >0$,
\item Type 3 utility function: 
$f_n(x)=\kappa_n {x}^{a_n}$ with $\kappa_n >0$ and \mbox{$0<a_n<1$.}

\end{itemize}

\begin{lemma}\label{lem-optPgivenB-type1}
For Type 1 utility function: 
$f_n(x)=\kappa_n \ln (1+ a_n x)$ with $\kappa_n >0$ and $0<a_n<1$, $\varphi_n(p_n,B_n)$ given $B_n$ achieves its maximum at $p_n$ given by the maximum of the following two numbers: $p_n^{\min} $ of Eq.~(\ref{eqdefpnmin}), and the solution $p_n$ to 
\begin{talign}
 & \textstyle{\ln (1+ a_n \cdot (B_n\log_2(1+\frac{g_n p_n}{{\sigma_n}^2B_n})- r_{n,\textnormal{e}}))} \nonumber \\ & =  \textstyle{W(\frac{a_ng_nB_n}{({\sigma_n}^2B_n+g_n p_n)\ln 2}  \cdot  (p_n + p_n^{\textnormal{cir}})  )},\label{eq-optPgivenB-type1}
\end{talign}
where $W(\cdot)$ is the principal branch of the Lambert W function.
\end{lemma} 



\begin{lemma} \label{lem-optPgivenB-type2}
For Type 2 utility function: 
\mbox{$f_n(x)= \kappa_n \cdot (1- e^{-a_n x})$} with $a_n, \kappa_n >0$, $\varphi_n(p_n,B_n)$ given $B_n$ achieves its maximum at $p_n$ given by  $$  \textstyle{\max\{p_n^{\min} \textnormal{ of Eq.~(\ref{eqdefpnmin})},~ (\chi_n-1)\frac{{\sigma_n}^2B_n}{g_n}\}},$$ where $\chi_n$ satisfies 
\begin{talign}
\textstyle{  {\chi_n}^{a_n B_n/(\ln2)}e^{- a_n r_{n,\textnormal{e}}} + \frac{  {\sigma_n}^2B_n -  g_n p_n^{\textnormal{cir}}}{\chi_n{\sigma_n}^2\ln 2} =  \frac{a_nB_n}{\ln 2}    + 1  }.\label{eq-optPgivenB-type2}
\end{talign}
\end{lemma}


\begin{lemma} \label{lem-optPgivenB-type3}
For Type 3 utility function: 
$f_n(x)=\kappa_n {x}^{a_n}$ with $\kappa_n >0$ and $0<a_n<1$, $\varphi_n(p_n,B_n)$ given $B_n$ achieves its maximum at $p_n$ being
\begin{talign}
&\textstyle{\max\hspace{-3pt}\bigg\{\hspace{-5pt}\begin{array}{l}p_n^{\min} \textnormal{}\\ \textnormal{of Eq.~(\ref{eqdefpnmin})},\end{array}\hspace{-2pt}\bigg[\hspace{-4pt}\textstyle{\begin{array}{l}
\exp\hspace{-2pt}\big(a_n  \hspace{-2pt} +\hspace{-2pt} \frac{r_{n,\textnormal{e}} \ln 2 }{B_n} + \\
W((\frac{a_ng_np_n^{\textnormal{cir}} }{{\sigma_n}^2B_n}\hspace{-2pt}-\hspace{-2pt}a_n)e^{ -a_n   - \frac{r_{n,\textnormal{e}} \ln 2 }{B_n} })\big)\hspace{-2pt} -\hspace{-2pt}1
\end{array}}\hspace{-5pt}\bigg]\hspace{-2pt}\frac{{\sigma_n}^2B_n}{g_n}}\hspace{-3pt}\bigg\}.   
\end{talign}
\end{lemma} 

The proofs of Lemmas~\ref{lemma-quasiconvex},~\ref{lem-optPgivenB-type1},~\ref{lem-optPgivenB-type2},~\ref{lem-optPgivenB-type3} are provided in
Section~\ref{appsec} below.



\subsection[]{Optimizing $\bm{B}$ only}\label{algo:B_n_only}


Given $p_n$, we define $B_n^{\min}$ as the value of $B_n$ which causes $r_n(p_n,B_n)$ to be $r_n^{\min}$; i.e., $B_n^{\min}\log_2(1+\frac{g_n p_n}{{\sigma_n}^2B_n^{\min}}) = r_n^{\min}$. 

Then we have the following optimization.
\begin{subequations} 
\begin{talign} 
 &\max_{B_n} ~  \sum_{n \in \mathcal{N}} c_n \varphi_n(p_n,B_n)
  \\ 
    \textnormal{subject to: }   & \sum_{n \in \mathcal{N}} B_n \le B_{\textnormal{total}}, \\ 
    &
    B_n \geq B_n^{\min}, \text{for any } n \in \mathcal{N}.
\end{talign}
\end{subequations}
From the first result of Lemma~\ref{lemma:fn_concave}, $f_n(r_{n,\textnormal{s}}(p_n,B_n))$ is concave with respect to $B_n$ given $p_n$. Hence, $\varphi_n(p_n,B_n)$ is concave with respect to $B_n$ given $p_n$. Then the above problem belongs to convex optimization.
The Lagrange function of the problem is as follows:
\begin{talign}\label{lagrange function:B_n}
    & L(p_n,B_n,\alpha_n,\zeta) \nonumber \\
    & = - \sum_{n \in \mathcal{N}} c_n\varphi_n(p_n,B_n) + \sum_{n \in \mathcal{N}} \alpha_n  \cdot( B_n^{\min}-B_n) \nonumber\\
    &+ \zeta \cdot (\sum_{n \in \mathcal{N}} B_n - B_{\textnormal{total}})
\end{talign}

After applying KKT conditions to (\ref{lagrange function:B_n}), we get: 
\begin{talign}
    &\frac{\partial L}{\partial p_n} = -c_n\nabla_{p_n} \varphi_n(p_n,B_n)= 0
     \\
    &\frac{\partial L}{\partial B_n} = -c_n\nabla_{B_n}
    \varphi_n(p_n,B_n) -\alpha_n + \zeta = 0 \label{lagrange function:B_n partial B_n}
    \\
    &\alpha_n \cdot ( B_n^{\min}-B_n) =0
     \\
    &\zeta \cdot (\sum_{n \in \mathcal{N}} B_n - B_{\textnormal{total}}) =0 \label{equa:zeta B_n}
\end{talign}







Since $\nabla_{B_n} \varphi_n(p_n,B_n) >0$ and $\alpha_n \geq 0$, only $\zeta > 0$ could make (\ref{lagrange function:B_n partial B_n}) hold. Hence from (\ref{equa:zeta B_n}), we could tell that $\sum_{n \in \mathcal{N}} B_n =B_{\textnormal{total}}$. Let $\widehat{B_n}(\zeta)$ be the solution of $B_n$ to the following equation:
\begin{align}\label{equa:diff B_n zeta}
c_n\nabla_{B_n} \varphi_n(p_n,B_n) = \zeta.
\end{align}
It is straightforward to derive the expression of $\widehat{B_n}(\zeta)$ based on (\ref{equa:diff B_n zeta}). 
For any $x>0$, we can prove $\log_2(1 + x) - \frac{x}{(1+x)\ln{2}} > 0$. Then we can prove that $\nabla_{B_n} \varphi_n(p_n,B_n) $ is decreasing as $B_n$ increases.
For each $n \in \mathcal{N}$, there are two possible cases:
\begin{itemize}
\item If $\alpha_n=0$, then $B_n \geq B_n^{\min}$, $c_n\nabla_{B_n} \varphi_n(p_n,B_n) = \zeta$ so that $B_n = \widehat{B_n}(\zeta) \geq B_n^{\min}$.
\item If $\alpha_n>0$, then $B_n = B_n^{\min}$, $c_n\nabla_{B_n} \varphi_n(p_n,B_n)  = \zeta - \alpha_n <  \zeta $.  Thus  $B_n = B_n^{\min} > \widehat{B_n}(\zeta)$.
\end{itemize}
Summarizing both cases and we can derive $B_n$ as follows: 
\begin{talign}\label{Eqn:widehat_B_n}
B_n = \max \{\widehat{B_n}(\zeta),~B_n^{\min}\}
\end{talign}
and $\zeta$  could be derived from:
\begin{talign}  
\sum_{n \in \mathcal{N}}\max \{\widehat{B_n}(\zeta),B_n^{\min}\} = B_{\text{total}}. \label{bnzeta}
\end{talign}

As $\nabla_{B_n} \varphi_n(p_n,B_n) $ is decreasing as $B_n$ increases. Then $\widehat{B_n}(\zeta)$ decreases as $\zeta$ increases. Let $\zeta^{\#}$ be the solution of  $\zeta$ to~(\ref{bnzeta}). We use the bisection search to find $\zeta^{\#}$. The following discussion is similar to that of Appendix~\ref{applambdastar}. We   let $F(\zeta)$ denote $\sum_{n \in \mathcal{N}}\max \{\widehat{B_n}(\zeta),B_n^{\min}\}$. Then $F(\zeta)$ is \mbox{non-increasing} as $\zeta$ increases. 

 For the  bisection method, we use $0$ as the initial lower bound. We can find the initial upper bound as follows. Starting with a random positive number $\theta$. If $F(\theta)$  is  less than $ B_{\textnormal{total}}$, we use $\theta$ as the initial upper bound. If $F(\theta)$ equals $ B_{\textnormal{total}}$, then $\theta$ is just our desired $\zeta^{\#}$. If $F(\theta)$  is  greater than $ B_{\textnormal{total}}$, we check $F(2\theta)$. Similarly, if $F(2\theta)$ is  less than $ B_{\textnormal{total}}$, we use $2\theta$ as the initial upper bound. If $F(2\theta)$  equals $ B_{\textnormal{total}}$, then $2\theta$ is just our desired $\zeta^{\#}$. If $F(2\theta)$ is  greater than $ B_{\textnormal{total}}$, we check $F(2^2\theta)$. The process continues. Basically, we find $i$ such that $F(2^{i-1}\theta)$ is  greater than $ B_{\textnormal{total}}$, and $F(2^i\theta)$ is  less than $ B_{\textnormal{total}}$, then we use $2^i\theta$ as the initial upper bound. If there exists $i$ which makes $F(2^i\theta)$ equal $ B_{\textnormal{total}}$, then $2^i\theta$ is just our desired $\zeta^{\#}$.

With the initial lower bound and the  initial upper bound explained above, the remaining process to  find $\zeta^{\#}$ follows from the standard bisection method. In each iteration, the bisection method divides the interval $[a,b]$ in two parts  by computing the midpoint $c = (a+b) / 2$ of the interval and the value of   $F(c)$. If $F(c)$ equals $ B_{\textnormal{total}}$, then the process has succeeded and $c$ is just our desired $\zeta^{\#}$. Otherwise, 
if $F(c)$ is greater than $ B_{\textnormal{total}}$, we update $a$ to $c$ so that the next iteration starts with the interval $[c,b]$; if $F(c)$ is less than $ B_{\textnormal{total}}$, we update $b$ to $c$ so that the next iteration starts with the interval $[a,c]$. The  bisection method   converges when $F(c)$ is close to (but should be no greater than) $ B_{\textnormal{total}}$.

After using the bisection method to find $\zeta^{\#}$, we compute $B_n$ as 
 $\max \{\widehat{B_n}(\zeta^{\#}),~B_n^{\min}\}$ for each $n \in \mathcal{N}$.

\subsection{Alternating optimization}

The algorithm for alternating optimization is to combine the algorithms of ``optimizing $\bm{p}$ only'' in \ref{algo:p_n_only} and ``optimizing $\bm{B}$ only'' in \ref{algo:B_n_only} to perform optimization in an alternating manner. 

Specifically, we treat ``optimizing $p_n$ only'' first and then ``optimizing $B_n$ only'' as a round. After each round, we will compare the new solution with that of the last round. If the relative difference between them is less than our pre-determined threshold, we consider that alternating optimization of $\bm{p}$ and $\bm{B}$ has converged.

\section{Proof of Lemmas for Appendix \ref{sec:baseline model}}\label{appsec}

\subsection{Proof of Lemma~\ref{lemma-quasiconvex}}


We first have the following properties of $\varphi_n(p_n,B_n)$.
\begin{itemize}
\item From Lemma~\ref{lemma:fn_concave}, the numerator $f_n(r_{n,\textnormal{s}}(p_n,B_n))$ of $\varphi_n(p_n,B_n)$ is jointly concave in $p_n$ and $B_n$. Then from Result (iii) of Theorem 2.3.8 in\footnote{Theorem 2.3.8 of~\cite{cambini2008generalized} and Lemma 2.2 of~\cite{ivanov2020characterization} are about ``semistrictly~quasiconvex'', but can easily be extended to ``semistrictly~quasiconcave'' since a function $g(\cdot)$ is semistrictly~quasiconvex  if and only if $-g(\cdot)$ is semistrictly~quasiconcave. Specifically, a function $g(\cdot)$ is semistrictly~quasiconvex (resp., semistrictly~quasiconcave) if and only if for any $x,y$, the result $g(y)<g(x)$ (resp., $g(y)>g(x)$) implies that $g(x+t(y-x))$ is smaller (resp., greater) than $g(x)$ for any $t\in(0,1)$. \label{quasiconcave} }~\cite{cambini2008generalized}, $\varphi_n(p_n,B_n)$ is semistrictly~quasiconcave which means that Result (i) of Lemma~\ref{lemma-quasiconvex} is  proved, where the definition of ``semistrictly~quasiconcave'' is provided in Footnote~\ref{quasiconcave}.
\item From Lemma 2.2 of~\cite{ivanov2020characterization}, a scalar function $g(\cdot)$ over a convex set $\mathcal{X}$ is semistrictly~quasiconcave if and only if any closed segment $\mathcal{S} \subset \mathcal{X}$ can be split into three segments $\mathcal{S}_1, \mathcal{S}_2, \mathcal{S}_3$ such that $g(\cdot)$ is increasing in $\mathcal{S}_1$,
constant in $\mathcal{S}_2$, and decreasing  in $\mathcal{S}_3$. Note that $\mathcal{S}_1, \mathcal{S}_2, \mathcal{S}_3$ can be $\emptyset$.
\end{itemize}
Then we define $p_n^{\secsec}$ as the transmission power as:
\begin{talign}
\text{$p_n^{\secsec}:=\frac{(2^{\frac{r_{n,\textnormal{e}}}{B_n}}-1){\sigma_n}^2B_n}{g_n}$ so that $B_n\log_2(1+\frac{g_n p_n^{\secsec}}{{\sigma_n}^2B_n}) = r_{n,\textnormal{e}}$.}
\end{talign}
Combining the above with Eq.~(\ref{eqdefvarphi}), we could derive that: 
$\varphi_n(p_n^{\secsec},B_n)=0$, $\lim_{p_n \to \infty}\varphi_n(p_n,B_n)=0$, and $\varphi_n(p_n,B_n)>0$ for $p_n \in (p_n^{\secsec},\infty)$, we know that\footnote{For the specific types of utility functions $f_n(\cdot)$ used in Section \ref{secTypicalUtility} of this paper, we can show $\widehat{p_n} =\widetilde{p_n}$ in principle, but it is not the focus of our paper and does not impact the validness of our results.} there exists $\widehat{p_n}$ and $\widetilde{p_n}$ (possibly the same) such that $\varphi_n(p_n,B_n)$ is increasing  in $[p_n^{\secsec},\widehat{p_n}]$,
constant in $[\widehat{p_n},\widetilde{p_n}]$, and  decreasing in $[\widetilde{p_n},\infty)$. 
 
The condition $r_n^{\min} \geq r_{n,\textnormal{e}}$ means $p_n^{\min} \geq p_n^{\secsec}$. Hence, the above analysis induces the following  cases.
\begin{itemize}
\item If $p_n^{\min} < \widehat{p_n}$, then $\varphi_n(p_n,B_n)$ is increasing in $[p_n^{\min},\widehat{p_n}]$,
constant in $[\widehat{p_n},\widetilde{p_n}]$, and  decreasing in $[\widetilde{p_n},\infty)$. 
\item If $\widehat{p_n} \leq p_n^{\min} \leq \widetilde{p_n}$, then $\varphi_n(p_n,B_n)$ is 
constant in $[\widehat{p_n},\widetilde{p_n}]$, and  decreasing in $[\widetilde{p_n},\infty)$.
\item If $\widetilde{p_n} < p_n^{\min}$, then $\varphi_n(p_n,B_n)$ is decreasing in $[p_n^{\min},\infty)$.
\end{itemize}
Based on the above, Result (ii) of Lemma~\ref{lemma-quasiconvex} is also proved. \qed



\subsection{Proof of Lemma~\ref{lem-optPgivenB-type1}}

Given Section~\ref{secTypicalUtility},   we consider $b_n=1$ so that Type 1 utility function is
$f_n(x)=\kappa_n \ln (1+ a_n x)$ with $\kappa_n >0$ and $0<a_n<1$.

We write $r_{n,\textnormal{s}}(p_n,B_n)$ as $r_{n,\textnormal{s}}$ for simplicity below.
Given $B_n$, by letting the derivative of $\varphi_n(p_n,B_n)$ with respect to $p_n$ be zero, we obtain
\begin{align}
\frac{\kappa_n a_n }{1+ a_n r_{n,\textnormal{s}}} \cdot (\nabla_{p_n} r_{n,\textnormal{s}}) \cdot (p_n + p_n^{\textnormal{cir}}) = \kappa_n \ln (1+ a_n r_{n,\textnormal{s}}) ,
\end{align}
which induces
\begin{align}
 a_n  \cdot (\nabla_{p_n} r_{n,\textnormal{s}}) \cdot (p_n + p_n^{\textnormal{cir}}) = (1+ a_n r_{n,\textnormal{s}})  \ln (1+ a_n r_{n,\textnormal{s}}) .
\end{align}
This means
\begin{align}
 \ln (1+ a_n r_{n,\textnormal{s}}) =  W(a_n  \cdot (\nabla_{p_n} r_{n,\textnormal{s}}) \cdot (p_n + p_n^{\textnormal{cir}}) ),
\end{align}
for $W(\cdot)$ being the principal branch of the Lambert W function; i.e., $W(z)$ for $z \geq -e^{-1}$ is the solution of $x \geq -1$ to the equation $x e^{x} = z$.

From $ r_{n,\textnormal{s}}(p_n,B_n):=B_n\log_2(1+\frac{g_n p_n}{{\sigma_n}^2B_n}) - r_{n,\textnormal{e}}  $, we have
\begin{align}
& (\nabla_{p_n} r_{n,\textnormal{s}}) \cdot (p_n + p_n^{\textnormal{cir}})\nonumber \\ & = \frac{\frac{g_n }{{\sigma_n}^2}}{(1+\frac{g_n p_n}{{\sigma_n}^2B_n})\ln 2}  \cdot  (p_n + p_n^{\textnormal{cir}}) \nonumber \\ & =  \frac{g_nB_n}{({\sigma_n}^2B_n+g_n p_n)\ln 2}  \cdot  (p_n + p_n^{\textnormal{cir}}). 
\end{align}
Then we can solve
\begin{align}
 &\ln (1+ a_n \cdot (B_n\log_2(1+\frac{g_n p_n}{{\sigma_n}^2B_n})- r_{n,\textnormal{e}})) \nonumber \\ & =  W(\frac{a_ng_nB_n}{({\sigma_n}^2B_n+g_n p_n)\ln 2}  \cdot  (p_n + p_n^{\textnormal{cir}})  ). \label{abd}
\end{align}
Using Result (ii) of Lemma~\ref{lemma-quasiconvex}, the optimal $p_n$ given $B_n$  for the  weighted sum-UEE optimization is given by the maximum of the following two numbers: $p_n^{\min} $ of Eq.~(\ref{eqdefpnmin}), and the solution $p_n$ to the above Eq.~(\ref{abd}). \qed

\subsection{Proof of Lemma~\ref{lem-optPgivenB-type2}}

 Given Section~\ref{secTypicalUtility},   we consider  $c_n=0$  so that   Type 2 utility function is
\mbox{$f_n(x)= \kappa_n \cdot (1- e^{-a_n x})$} with $a_n, \kappa_n >0$.

We write $r_{n,\textnormal{s}}(p_n,B_n)$ as $r_{n,\textnormal{s}}$ for simplicity below.
Given $B_n$, by letting the derivative of $\varphi_n(p_n,B_n)$ with respect to $p_n$ be zero, we obtain
\begin{align}
\kappa_na_ne^{-a_n r_{n,\textnormal{s}}} (\nabla_{p_n} r_{n,\textnormal{s}}) \cdot (p_n + p_n^{\textnormal{cir}}) = \kappa_n \cdot  (1-e^{-a_n r_{n,\textnormal{s}}} ),
\end{align}
which induces
\begin{align}
a_n\cdot (\nabla_{p_n} r_{n,\textnormal{s}}) \cdot (p_n + p_n^{\textnormal{cir}}) = e^{a_n r_{n,\textnormal{s}}}-1 .
\end{align}

With
$\chi_n: = 1+\frac{g_n p_n}{{\sigma_n}^2B_n}$,
we further obtain
\begin{align}
& (\nabla_{p_n} r_{n,\textnormal{s}}) \cdot (p_n + p_n^{\textnormal{cir}})\nonumber \\ & = \frac{\frac{g_n }{{\sigma_n}^2}}{(1+\frac{g_n p_n}{{\sigma_n}^2B_n})\ln 2}  \cdot  (p_n + p_n^{\textnormal{cir}}) \nonumber \\ & = \frac{ (\chi_n-1){\sigma_n}^2B_n + g_n p_n^{\textnormal{cir}}}{\chi_n{\sigma_n}^2\ln 2}  \nonumber \\ & = \frac{B_n}{\ln 2} +  \frac{    g_n p_n^{\textnormal{cir}}-{\sigma_n}^2B_n}{\chi_n{\sigma_n}^2\ln 2} .
\end{align}
Then
\begin{align}
&   a_n\cdot (\frac{B_n}{\ln 2} +  \frac{    g_n p_n^{\textnormal{cir}}-{\sigma_n}^2B_n}{\chi_n{\sigma_n}^2\ln 2} )  \nonumber \\ & = e^{a_n B_n\log_2 \chi_n} e^{- a_n r_{n,\textnormal{e}}}- 1\nonumber \\ & = {\chi_n}^{a_n B_n/(\ln2)}e^{- a_n r_{n,\textnormal{e}}}-1.
\end{align}
We further get
\begin{align}
&    {\chi_n}^{a_n B_n/(\ln2)}e^{- a_n r_{n,\textnormal{e}}} + \frac{  {\sigma_n}^2B_n -  g_n p_n^{\textnormal{cir}}}{\chi_n{\sigma_n}^2\ln 2} = a_n\frac{B_n}{\ln 2}    + 1.  \label{abd2} 
\end{align}
Using Result (ii) of Lemma~\ref{lemma-quasiconvex}, the optimal $p_n$ given $B_n$  for the  weighted sum-UEE optimization is given by the maximum of the following two numbers: $p_n^{\min} $ of Eq.~(\ref{eqdefpnmin}), and $(\chi_n-1)\frac{{\sigma_n}^2B_n}{g_n}$, where $\chi_n$ is the solution  to the above Eq.~(\ref{abd2}). \qed

\subsection{Proof of Lemma~\ref{lem-optPgivenB-type3}}

Given Section~\ref{secTypicalUtility},   we consider  $d_n=0$     so that  Type 3 utility function is
$f_n(x)=\kappa_n {x}^{a_n}$ with $\kappa_n >0$ and \mbox{$0<a_n<1$.}

We write $r_{n,\textnormal{s}}(p_n,B_n)$ as $r_{n,\textnormal{s}}$ for simplicity below.
Given $B_n$, by letting the derivative of $\varphi_n(p_n,B_n)$ with respect to $p_n$ be zero, we obtain
\begin{align}
\kappa_n a_n {r_{n,\textnormal{s}}}^{a_n-1}\cdot (\nabla_{p_n} r_{n,\textnormal{s}}) \cdot (p_n + p_n^{\textnormal{cir}}) = \kappa_n {r_{n,\textnormal{s}}}^{a_n} .  
\end{align}
Then we have
\begin{align}
 a_n  \cdot (\nabla_{p_n} r_{n,\textnormal{s}}) \cdot (p_n + p_n^{\textnormal{cir}}) = {r_{n,\textnormal{s}}} . \label{sbsb}
\end{align}

From $ r_{n,\textnormal{s}}(p_n,B_n):=B_n\log_2(1+\frac{g_n p_n}{{\sigma_n}^2B_n}) - r_{n,\textnormal{e}}  $, we have
    \begin{align}
  & \nabla_{p_n} r_{n,\textnormal{s}}   =  \frac{\frac{g_n }{{\sigma_n}^2}}{(1+\frac{g_n p_n}{{\sigma_n}^2B_n})\ln 2},   
    \end{align}
which is used in~(\ref{sbsb}) so that 
\begin{align}
 a_n \cdot \frac{\frac{g_n }{{\sigma_n}^2}}{(1+\frac{g_n p_n}{{\sigma_n}^2B_n})\ln 2}  \cdot  (p_n + p_n^{\textnormal{cir}}) = B_n\log_2(1+\frac{g_n p_n}{{\sigma_n}^2B_n})- r_{n,\textnormal{e}}.
\end{align}


With
$\chi_n: = 1+\frac{g_n p_n}{{\sigma_n}^2B_n}$,
we further obtain
\begin{align}
 a_n \cdot \frac{\frac{g_n }{{\sigma_n}^2}}{\chi_n}  ((\chi_n-1)\frac{{\sigma_n}^2B_n}{g_n} + p_n^{\textnormal{cir}}) + r_{n,\textnormal{e}} \ln 2 = B_n\ln\chi_n, 
\end{align}
which is simplified as
\begin{align}
 a_n   + \frac{r_{n,\textnormal{e}} \ln 2 }{B_n} + (\frac{a_ng_np_n^{\textnormal{cir}} }{{\sigma_n}^2B_n}-a_n )\frac{1}{\chi_n}   =  \ln\chi_n .
\end{align}
Then it follows that
\begin{align}
\chi_n   = \textstyle{\exp\big(a_n   + \frac{r_{n,\textnormal{e}} \ln 2 }{B_n} + (\frac{a_ng_np_n^{\textnormal{cir}} }{{\sigma_n}^2B_n}-a_n)\exp( -a_n   - \frac{r_{n,\textnormal{e}} \ln 2 }{B_n} )\big)}  . \label{abd3}
\end{align} 
Using Result (ii) of Lemma~\ref{lemma-quasiconvex}, the optimal $p_n$ given $B_n$  for the  weighted sum-UEE optimization is given by the maximum of the following two numbers: $p_n^{\min} $ of Eq.~(\ref{eqdefpnmin}), and $(\chi_n-1)\frac{{\sigma_n}^2B_n}{g_n}$, where $\chi_n$ is given by the above Eq.~(\ref{abd3}). \qed


\section{Additional Simulation Results} \label{appAdditional}

We provide more simulation results here. Note that our theoretical analysis and simulation results  apply to uplink communications as well as downlink communications. 
For downlink communications, the circuit power $p_n^{\textnormal{cir}}$ in the denominator $p_n + p_n^{\textnormal{cir}}$ of (\ref{eqdefvarphi}) is the additional power that the server consumes to transmit signals with power $p_n$ to user $U_n$, and it is possible that $p_n^{\textnormal{cir}}$ can be the same for different $n$ for such downlink communications. As noted in Section \ref{experiment:parameter}, the simulations set the circuit power $p_n^{cir}$ as 2\,dBm (i.e., 1.6\,milliwatts) for each $n$.




\subsection{Impact by the number of users}\label{appImpactnumberusers}
We  compare the UEE under different numbers of users $N$: 10, 20, 30, 40, and 50 under three different user scenarios (i.e., three utility functions in Fig. \ref{fig:real_simul}(a)). Fig. \ref{fig:N}(a) shows how the sum-UEE changes with the increase of the number of users $N$. It could be seen that as $N$ becomes larger, the sum-UEE also increases. They are positively correlated. In contrast, the average UEE tends to decrease as the number of users $N$ increases, which could be seen in Fig.~\ref{fig:N}(b). That is due to a reduction in the bandwidth  allocated to each user. 

\begin{figure}[ht!]
    \centering
    \includegraphics[width=0.9\linewidth]{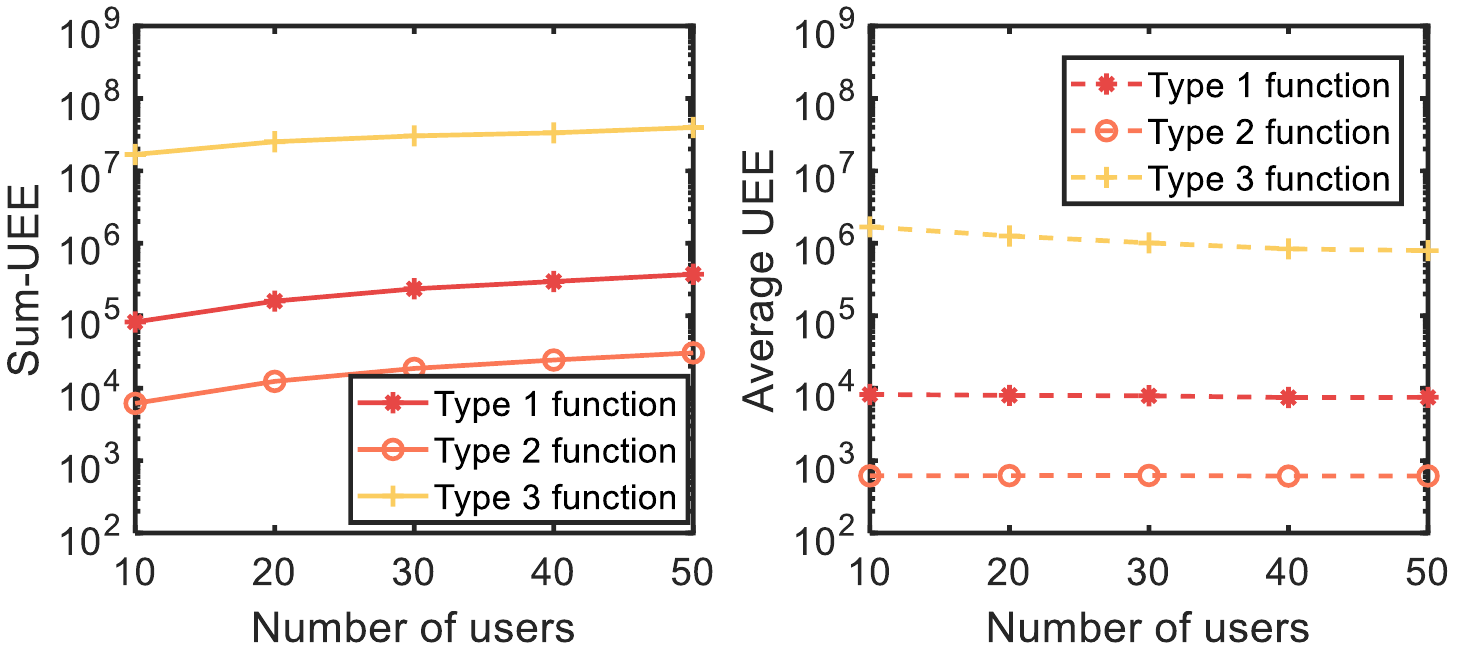}
    \caption{(a). Sum-UEE under different $N$. 
    (b). Average UEE under different $N$.
    }
    \vspace{-10pt}
    \label{fig:N}
\end{figure}

\begin{figure}[ht!]
    \centering
    \includegraphics[width=0.3\textwidth]{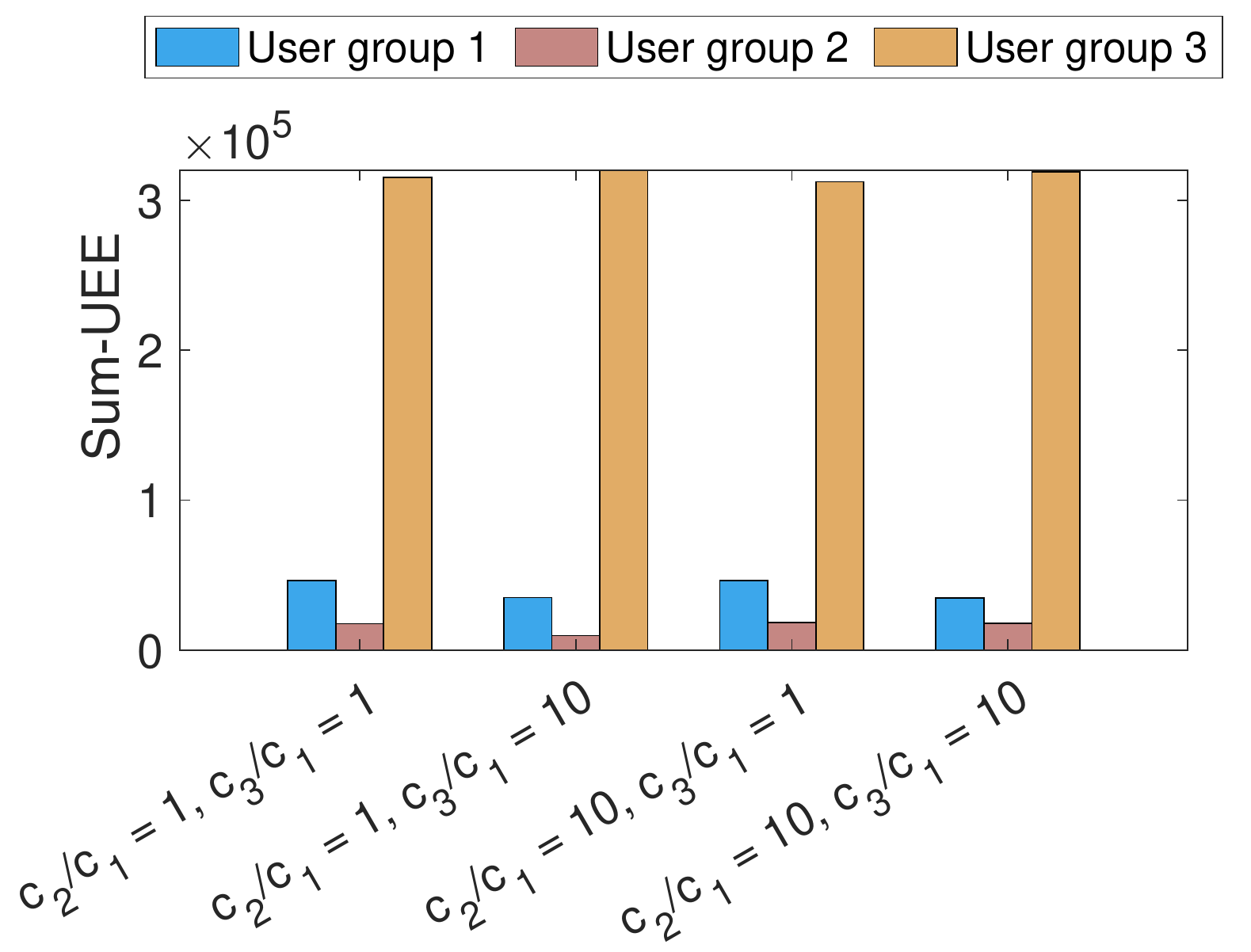}
    \caption{Sum-UEE of each user group with different priorities and utility functions. The utility functions are obtained from the table results in Section~\ref{realdata} with the SSV360  dataset~\cite{ssv360}. User group 1 uses $f_n(r_{n,\textnormal{s}})= 0.5424 \ln (1+ 37.2965 r_{n,\textnormal{s}})$. User group 2 uses $f_n(r_{n,\textnormal{s}})= 2.9351 (1- e^{-2.1224 r_{n,\textnormal{s}}})$. User group~3 uses $f_n(r_{n,\textnormal{s}})= 3.2956 {(r_{n,\textnormal{s}}/15.94)}^{0.2733}$. }
    \label{fig:uee_diff_types}
\end{figure}

\subsection{Heterogeneous types of utility functions among the users} \label{secHeterogeneous}

We use Figure~\ref{fig:uee_diff_types} to show that our studied system and proposed algorithm  allow heterogeneous types of utility functions among the users. We consider that $30$ users are evenly classified into three priority levels, corresponding to different weights $c_n$. Larger $c_n$ means more weight in our studied optimization.  
For example, the legend ``$c_2/c_1 = a, c_3/c_1 =b$'' in Fig.~\ref{fig:uee_diff_types} means  $10$ users in Group 1 with weight $c_1$, $10$ users in Group 2 with weight $a c_1$, and $10$ users in Group 3 with weight $b c_1$. 
 In Figure~\ref{fig:uee_diff_types}, we can see that with a fixed $c_2/c_1 $, increasing $c_3/c_1$ will improve  Group 3's sum-UEE, and reduce Group 1's and 2's sum-UEE; with a fixed $c_3/c_1 $, enlarging $c_2/c_1$ will enhance  Group 2's sum-UEE, and shrink Group 1's and~3's sum-UEE. The above simulation results are consistent with the intuition.

\section{Using our technique for global optimization of fractional programming} \label{secfractionalprogramming}

 Two recent papers~\cite{shen2018fractional,shen2018fractional2} by Shen and Yu are well-cited and have been considered breakthroughs in fractional programming. However, their proposed technique finds neither locally nor globally optimal solution. In contrast, with our technique of Section~\ref{Insights}, a globally optimal solution can be found.  The following problem is considered by Shen and Yu~\cite{shen2018fractional,shen2018fractional2}.

The following problem $\mathbb{P}_4$ is considered by Shen and Yu~\cite{shen2018fractional,shen2018fractional2}, where $A_n(\cdot)$, $B_n(\cdot)$, $C(\cdot)$, $g_m(\cdot)$, and $h_{\ell}(\cdot)$ are functions, with $A_n(\boldsymbol{x})>0$ and $B_n(\boldsymbol{x})>0$ for all $n=1,2,\ldots,N$.
\begin{subequations}\label{problem:4}
\begin{align}
   \textnormal{Problem $\mathbb{P}_4$:}  ~ &\min_{\boldsymbol{x}\in\mathbb{R}^J}~C(\boldsymbol{x}) +\sum_{n=1}^N\frac{A_n(\boldsymbol{x})}{B_n(\boldsymbol{x})}
    \tag{\ref{problem:4}}\\ 
   &  \textnormal{subject to: }   g_m(\boldsymbol{x}) \leq 0,\textnormal{ for $m=1,2,\ldots,M$},  \label{constra:gm} \\ &\hspace{55pt} h_{\ell}(\boldsymbol{x}) = 0,\textnormal{ for $\ell=1,2,\ldots,L$}.\label{constra:hell}
\end{align}
\end{subequations}

We introduce an auxiliary variable $\alpha_n$ to transform Problem $\mathbb{P}_{4}$ into the epigraph form. Let $\alpha_n \ge \frac{A_n(\boldsymbol{x})}{B_n(\boldsymbol{x})}$ and $\mathbb{P}_{4}$ can be transformed to the following equivalent form as $\mathbb{P}_5$:
\begin{subequations}\label{problem:5}
\begin{align}
   \textnormal{Problem $\mathbb{P}_5$:}  ~ &\min_{\boldsymbol{x}\in\mathbb{R}^J,\, \boldsymbol{\alpha}\in\mathbb{R}^N}~C(\boldsymbol{x}) +\sum_{n=1}^N\alpha_n
    \tag{\ref{problem:5}}\\ 
   &  \textnormal{subject to: }   A_n(\boldsymbol{x}) -\alpha_n B_n(\boldsymbol{x}) \leq  0,\textnormal{ for   $n=1,2,\ldots,N$},  \\ 
   & \hspace{55pt}g_m(\boldsymbol{x}) \leq 0,\textnormal{ for $m=1,2,\ldots,M$},  \label{constra:gm} \\ &\hspace{55pt} h_{\ell}(\boldsymbol{x}) = 0,\textnormal{ for $\ell=1,2,\ldots,L$},\label{constra:hell}
\end{align}
\end{subequations}
Similar to how we connect $\mathbb{P}_{2}$ and $\mathbb{P}_3(\bm{\beta},\bm{\nu})$ in Section~\ref{secTransforming}, we can connect $\mathbb{P}_{5}$ and $\mathbb{P}_6(\boldsymbol{\alpha},\boldsymbol{\beta})$ defined as follows:
\begin{subequations}\label{problem:6}
\begin{align}
  & \textnormal{Problem $\mathbb{P}_6(\boldsymbol{\alpha},\boldsymbol{\beta})$:} \nonumber \\ 
   &  \min_{\boldsymbol{x}\in\mathbb{R}^J}~C(\boldsymbol{x}) +\sum_{n=1}^N \alpha_n  + \sum_{n=1}^N\beta_{n} \cdot (A_n(\boldsymbol{x}) -\alpha_n B_n(\boldsymbol{x}) )
    \tag{\ref{problem:6}}\\ 
   &  \textnormal{subject to: }    g_m(\boldsymbol{x}) \leq 0,\textnormal{ for $m=1,2,\ldots,M$},  \label{constra:gm} \\ &\hspace{55pt} h_{\ell}(\boldsymbol{x}) = 0,\textnormal{ for $\ell=1,2,\ldots,L$},\label{constra:hell}
\end{align}
\end{subequations}

If $A_n(\cdot)$, $C(\cdot)$, $g_m(\cdot)$ are convex, $B_n(\cdot)$ is concave, and $h_{\ell}$ is affine, then $\mathbb{P}_6(\boldsymbol{\alpha},\boldsymbol{\beta})$ belongs to convex optimization.

The solving process of $\mathbb{P}_5$ (i.e., $\mathbb{P}_4$) is transformed into solving a series of parametric convex optimization $\mathbb{P}_6(\boldsymbol{\alpha},\boldsymbol{\beta})$ where $[\boldsymbol{\alpha},\boldsymbol{\beta}]$ is given so that there is no \mbox{non-convex} product term $\alpha_n B_n(\boldsymbol{x})$. The solving of each $\mathbb{P}_6$ is used to update $[\boldsymbol{\alpha},\boldsymbol{\beta}]$ under which $\mathbb{P}_6$ is solved again with the new  
$[\boldsymbol{\alpha},\boldsymbol{\beta}]$, where the update of $[\boldsymbol{\alpha},\boldsymbol{\beta}]$ is based on the KKT conditions of $\mathbb{P}_5$. The process is similar to what we have presented in Section~\ref{secTransforming} for $\mathbb{P}_{2}$ and $\mathbb{P}_3(\bm{\beta},\bm{\nu})$. 

\end{appendix}

\end{document}